\documentclass{sigplanconf}
\usepackage{url}
\usepackage{tikz}
\usepackage{amsmath}
\usepackage{amssymb}
\usepackage{amsthm}
\usepackage{stmaryrd}
\usepackage{mathpartir}

\makeatletter
\toappear{%
  \noindent \@permission \par
  \vspace{2pt}
  \noindent \textsl{\@conferencename}\quad \@conferenceinfo
}
\makeatother

\usepackage{mathptmx}
\DeclareSymbolFont{letters}{OML}{txmi}{m}{it}
\DeclareMathAlphabet{\mathcal}{OMS}{cmsy}{m}{n}

\newcommand{\hsapp}{\mathbin{+\mkern-7mu+}}
\newcommand{\hsbind}{\mathbin{>\mkern-6mu>\mkern-6.5mu=}}
\newcommand{\To}{\mathbin{\Rightarrow}}

\newcommand{\U}{\ensuremath{\mathcal{U}}}
\newcommand{\D}{\ensuremath{\mathcal{D}}}
\newcommand{\below}{\sqsubseteq}
\newcommand{\univ}[1]{\ensuremath{\underline{#1}}}
\newcommand{\REP}[1]{\ensuremath{\llbracket#1\rrbracket}}
\newcommand{\symlbrace}{\mathopen{\lbrace\mkern-4.5mu\mid}}
\newcommand{\symrbrace}{\mathclose{\mid\mkern-4.5mu\rbrace}}
\newcommand{\TC}[1]{\ensuremath{\symlbrace#1\symrbrace}}
\newcommand{\hsone}{\mathbf{1}}
\newcommand{\hair}{\mskip1mu}
\newcommand{\mempty}{\varnothing}
\newcommand{\mappend}{\mathrel{\bullet}}

\newcommand{\kwd}[1]{\mathbf{#1}}

\newcommand{\hsc}[1]{\ensuremath{\mathit{#1}}}

\newcommand{\hsid}{\hsc{id}}
\newcommand{\hsemb}{\hsc{emb}}
\newcommand{\hsprj}{\hsc{proj}}
\newcommand{\hsproj}{\hsc{proj}}
\newcommand{\hsRep}{\hsc{Rep}}
\newcommand{\hscoerce}{\hsc{coerce}}
\newcommand{\fmap}{\,\hsc{fmap}} %mathptmx
\newcommand{\fmapU}{\univ{\fmap}}
\newcommand{\returnU}{\univ{\hsc{return}}}
\newcommand{\hsbindU}{\mathbin{\univ{\hsbind}}}
\newcommand{\hsappU}{\mathbin{\univ{\hsapp}}}
\newcommand{\runET}{\hsc{runET}}

\newcommand{\mapLift}{\hsc{map_\bot}}
\newcommand{\mapSum}{\hsc{map_\oplus}}
\newcommand{\mapProd}{\hsc{map_\otimes}}
\newcommand{\mapFun}{\hsc{map_\to}}

\newcommand{\embedding}{\hsc{in}}
\newcommand{\projection}{\hsc{out}}

\newcommand{\INV}{\text{\textsc{Inv}}}

\newcommand{\tA}{\alpha}
\newcommand{\tB}{\beta}
\newcommand{\tC}{\gamma}
\newcommand{\tD}{\delta}
\newcommand{\tE}{\varepsilon}
\newcommand{\tT}{\tau}
\newcommand{\tW}{\omega}

\newcommand{\isodefl}{\Vdash}
\newcommand{\defeq}{\stackrel{\mathrm{def}}{=}}
\newcommand{\eppair}{\stackrel{ep}{\to}}

\newcommand{\justification}[1]{\{\text{#1}\}}

\newtheorem{theorem}{Theorem}
\theoremstyle{definition}
\newtheorem{definition}{Definition}

\begin{document}

\conferenceinfo{ICFP'12,} {September 9--15, 2012, Copenhagen, Denmark.}
\copyrightyear{2012}
\copyrightdata{978-1-4503-1054-3/12/09}
\permission{\vtop{\vskip-1.5ex\hbox{\includegraphics[scale=0.75]{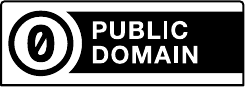}}} \hfill \parbox[t]{6.45cm}{To the extent possible under law, Brian Huffman has waived all copyright and related or neighboring rights to ``Formal Verification of Monad Transformers''. This work is published from: Germany.}}

%\titlebanner{DRAFT --- Do not distribute}

\title{Formal Verification of Monad Transformers}

\authorinfo{Brian Huffman}{Institut f\"ur Informatik, Technische Universit\"at M\"unchen}{huffman@in.tum.de}

\maketitle

%%%%%%%%%%%%%%%%%%%%%%%%%%%%%%%%%%%%%%%%%%%%%%%%%%
\begin{abstract}
We present techniques for reasoning about constructor classes that (like the monad class) fix polymorphic operations and assert polymorphic axioms. We do not require a logic with first-class type constructors, first-class polymorphism, or type quantification; instead, we rely on a domain-theoretic model of the type system in a universal domain to provide these features.

These ideas are implemented in the Tycon library for the Isabelle theorem prover, which builds on the HOLCF library of domain theory. The Tycon library provides various axiomatic type constructor classes, including functors and monads. It also provides automation for instantiating those classes, and for defining further subclasses.

We use the Tycon library to formalize three Haskell monad transformers: the error transformer, the writer transformer, and the resumption transformer. The error and writer transformers do not universally preserve the monad laws; however, we establish datatype invariants for each, showing that they are valid monads when viewed as abstract datatypes.
\end{abstract}



%%%%%%%%%%%%%%%%%%%%%%%%%%%%%%%%%%%%%%%%%%%%%%%%%%
\section{Introduction}

As a pure functional language, Haskell promises to work well for equational reasoning and proofs. Having programs and libraries that satisfy equational laws is important, because it lets programmers think about the correctness of their code in a modular and composable way.

Type classes are a valuable abstraction mechanism for writing reusable code in Haskell. Many Haskell type classes also have laws associated with them. Haskell programs that use these type classes often rely on the assumption that the laws hold. For example, a library might implement a datatype of balanced search trees, with elements of type $\tA$. To permit comparisons between elements, the search tree operations use the class constraint $\hsc{Ord}\:\tA$, which provides the comparison operator $(\le)::\tA\to\tA\to\hsc{Bool}$. But just having an operation of the right type is not enough: For the operations to work correctly, the library requires $(\le)$ to satisfy some additional properties, e.g. that $(\le)$ is a \emph{total} order.

Much Haskell code is written with equational properties in mind: Programs, libraries, and class instances may be expected to satisfy some laws, but unfortunately, there is no formal connection between programs and properties in Haskell. Haskell compilers are not able to check that properties hold. One way to get around this limitation is to verify our Haskell programs in an interactive proof assistant, or theorem prover.

\paragraph{Isabelle/HOL.}

Isabelle/HOL (or simply ``Isabelle'') is a generic interactive theorem prover, with tools and automation for reasoning about inductive datatypes and terminating functions in higher-order logic \cite{isabelle-tutorial}. Isabelle has an ML-like type system extended with \emph{axiomatic} type classes \cite{Wenzel1997}. In Isabelle, a type class fixes one or more overloaded constants, just like in Haskell. But Isabelle also allows us to specify additional \emph{class axioms} about those constants.

As an example, here we have an axiomatic class \hsc{Ord} that fixes an order relation $(\le)$ and asserts that it is a total order:
\begin{align*}
  & \kwd{class}\:\hsc{Ord}\:\tA\:\kwd{where} \\[-\jot]
  & \hspace{8pt} (\le)::\tA\to\tA\to\hsc{Bool} \\
  & \hspace{8pt}
  \begin{aligned}
  & %\forall\,{x},
  {x}\le{x} \\[-\jot]
  & %\forall\,{x}\,{y}\,{z},
  {x}\le{y}\;\wedge\; {y}\le{z}\implies{x}\le{z} \\[-\jot]
  & %\forall\,{x}\,{y},
  {x}\le{y} \;\wedge\; {y}\le{x}\implies{x}={y} \\[-\jot]
  & %\forall\,{x}\,{y},
  {x}\le{y} \;\vee\; {y}\le{x}
  \end{aligned}
\end{align*}
(Note that free variables appearing in class axioms are treated as universally quantified.) To establish an instance of class \hsc{Ord} in Isabelle, a user must not only provide definitions of the class operations, but also proofs that the operations satisfy the class axioms.

\paragraph{Isabelle/HOLCF.}

Isabelle/HOLCF is a library of domain theory, formalized within the logic of Isabelle/HOL \cite{holcf99,holcf11}. It is designed to support denotational reasoning about programs written in pure functional languages like Haskell. HOLCF can deal with programs that are beyond the scope of Isabelle/HOL's automation: HOLCF provides tools for defining and working with (possibly lazy) recursive datatypes, general recursive functions, partial and infinite values, and least fixed-points. These features make Isabelle/HOLCF a useful system for reasoning about a significant subset of Haskell programs. With the combination of HOLCF and axiomatic classes, users can directly formalize many Haskell programs that use ad hoc overloading, and verify generic programs that may rely on laws for class instances.

\paragraph{Type constructor classes.}

In addition to ordinary type classes, Haskell also supports type \emph{constructor} classes. An ordinary type constraint like $\hsc{Ord}\:\tA$ involves a type variable $\tA :: *$. The operations in such a type class have relatively simple types like $(\le) :: (\hsc{Ord}\:\tA) \To \tA \to \tA \to \hsc{Bool}$, where no other type variables besides the one in the class constraint are mentioned. That is, for a specific class instance, the operations are monomorphic: e.g. to define an instance $\hsc{Ord}\;\hsc{Int}$, we have an operation $(\le^\hsc{Int}) :: \hsc{Int} \to \hsc{Int} \to \hsc{Bool}$. (In other words, a dictionary for class \hsc{Ord} contains only monomorphic functions.) On the other hand, a constructor class like $\hsc{Functor}\:\tT$ fixes a type variable of a higher kind, in this case $\tT :: * \to *$. Furthermore, the operations in a constructor class are usually polymorphic. For example, $\fmap :: (\hsc{Functor}\:\tT) \To (\tA \to \tB) \to \tT\:\tA \to \tT\:\tB$ also is polymorphic over the type variables $\tA$ and $\tB$. The laws for the functor class are likewise polymorphic: For functors we usually assume that $\fmap\;\hsid = \hsid$ and $\fmap\;(f \circ g) = \fmap\;f \circ \fmap\;g$. For a specific functor class instance, these laws can be instantiated at various types. For a proper, law-abiding functor, we expect these laws to hold at \emph{all} possible type instantiations.

These additional requirements pose some real challenges for formal verification. While Isabelle has built-in support for ordinary axiomatic type classes, its type system does not natively support axiomatic constructor classes or type quantification---in fact, it does not even support higher-kinded type variables at all. Other interactive theorem provers exist with stronger type systems (e.g. Coq), but switching would mean giving up all the special support for reasoning about strictness, partial values, and general recursion in HOLCF. Coq and similar provers use a logic of total, terminating functions; thus proofs conducted in them are only applicable to the total, terminating fragment of the Haskell language.

\paragraph{Contributions.}

Using a universal domain and a domain-theoretic model of types, we construct a library for Isabelle/HOLCF that gives users first-class type constructors and axiomatic constructor classes. Users can instantiate constructor classes by defining the constants and proving the class axioms at a single type. Using a combination of type coercions and naturality laws, theorems can then be transferred automatically to other type instances.

This work builds upon and improves an earlier formalization of constructor classes in Isabelle, which was joint work with Matthews and White \cite{HMW2005}. While some concepts (e.g. representable types, the type application operator, and coercions) remain unchanged, this paper also introduces several new contributions:

\begin{itemize}
\item New simplified definition of class \hsc{Functor}
\item Fully automatic tools for constructing \hsc{Functor} class instances
\item A general, practical method for defining subclasses of \hsc{Functor}
\item Automation for transferring theorems between types, using coercions and naturality laws
\item Verification of error and writer monad transformers as abstract datatypes
\end{itemize}

\paragraph{Outline.}

The remainder of the paper is organized as follows. We begin by reviewing relevant information about HOLCF: After a summary of basic domain-theoretic concepts (\S\ref{sec:holcf}), we discuss the deflation model used to represent types in HOLCF (\S\ref{sec:deflation-model}). The next sections cover the implementation of the Tycon library: We show how to define the various constructor classes (\S\ref{sec:constructor-classes}), and then how to instantiate them (\S\ref{sec:instantiation}). Next we discuss the verification of monad transformers with Tycon (\S\ref{sec:monad-transformers}). Finally, we conclude with a discussion of related and future work (\S\ref{sec:conclusion}).

%%%%%%%%%%%%%%%%%%%%%%%%%%%%%%%%%%%%%%%%%%%%%%%%%%
\section{Domain theory in HOLCF}
\label{sec:holcf}

We now review the basic domain theory definitions used in HOLCF. A partial order is a set or type with a binary relation $(\below)$ that is reflexive, transitive, and antisymmetric. A chain is a countable increasing sequence: $x_0 \below x_1 \below x_2 \below \dots$ A complete partial order (cpo) is a partial order where every chain has a least upper bound (lub). An admissible predicate $P$ holds for the lub of a chain whenever it holds over the entire chain: $\forall{n}.\,P(x_n) \Longrightarrow P\left(\bigsqcup_n x_n\right)$. A continuous function $f$ preserves lubs of chains: $f\left(\bigsqcup_n x_n\right) = \bigsqcup_n f(x_n)$. Note also that every continuous function is monotone. A pointed cpo (pcpo) or ``domain'' is a cpo with a least element $\bot$. Every continuous function $f$ on a pcpo has a least fixed-point $\mathit{fix}(f) = f(\mathit{fix}(f)) = \bigsqcup_n f^n(\bot)$. In this paper we also use the binder notation $\mu\hair{x}.\,f(x)$ to denote the least fixed-point of $f$.

HOLCF provides a few primitive type constructors, which correspond to basic domain constructions. First, we have the continuous function space $\tA\to\tB$, which consists of the continuous functions from $\tA$ to $\tB$ ordered pointwise; this type is used to model Haskell's function space. Other constructions include strict sums, strict products, and lifting. They correspond to the Haskell datatype definitions here:
\begin{align*}
& \kwd{data}\:\tA\oplus\tB = \hsc{SLeft}\:!\tA \mid \hsc{SRight}\:!\tB \\
& \kwd{data}\:\tA\otimes\tB = \hsc{SPair}\:!\tA\:!\tB \\
& \kwd{data}\:\tA_\bot = \hsc{Lift}\:\tA
\end{align*}
Note that the constructors for $\tA\oplus\tB$ and $\tA\otimes\tB$ are strict, but the constructor for type $\tA_\bot$ is non-strict: $\hsc{Lift}\:\bot \neq \bot$. Finally, HOLCF provides the type $\hsone$, with two elements $\bot \below ()$; this models Haskell's unit type $()$.

\subsection{The Domain package}

Constructing recursive datatypes is one important application of domain theory. In HOLCF, user-defined recursive datatypes can be specified using the Domain package \cite{holcf11}. It can model many of the same datatypes that we can define in Haskell, e.g., lazy lists:
\begin{equation*}
\kwd{data}\:List\:\tA = \hsc{Nil} \mid \hsc{Cons}\:\tA\:(\hsc{List}\:\tA)
\end{equation*}
Given this datatype specification, the job of the Domain package is to construct a solution to the corresponding domain equation.
\begin{equation}
\hsc{List}\:\tA \cong \hsone \oplus (\tA_\bot \otimes (\hsc{List}\:\tA)_\bot)
\end{equation}

Since Isabelle 2011, the Domain package is completely definitional: It explicitly constructs a solution to this equation and defines $\hsc{List}\:\tA$ without introducing any new axioms \cite{holcf11}. In addition to the type itself, the Domain package also defines the constructor functions and several other related constants. It also generates a large collection of useful lemmas and rewrite rules, including injectivity and exhaustiveness of constructors, and rules for order comparisons like this one:
\begin{equation}
\hsc{Cons}\;x\;xs \below \hsc{Cons}\;y\;ys \Longleftrightarrow x \below y \wedge xs \below ys
\end{equation}

We also get some induction rules generated for us: Every type gets a low-level induction principle in the form of an approximation lemma \cite{Hutton2001}. For polynomial types (i.e., those expressible as a sum of products) we also get a high-level induction rule, with cases for each constructor plus a case for $\bot$. Induction rules for lazy datatypes have an admissibility side-condition.
\begin{equation}
\inferrule
{\mathrm{admissible}(P) \\ P(\bot) \\ P(\hsc{Nil}) \\
  \forall{x}\;\hsc{xs}.\;P(\hsc{xs}) \Longrightarrow P(\hsc{Cons}\;x\;\hsc{xs})}
{\forall\hsc{xs}.\;P(\hsc{xs})}
\end{equation}

The Domain package can handle any datatype expressible in Haskell---subject to the limitations of Isabelle's type system, of course. It supports both strict and lazy constructors, mutual recursion, indirect recursion, and even negative recursion.
\begin{align*}
  & \kwd{data}\;\hsc{StrictList}\;\tA =
  \hsc{SNil} \mid \hsc{SCons}\;\tA\;!(\hsc{StrictList}\;\tA) \\
  & \kwd{data}\;\hsc{Indirect}\;\tA =
  \hsc{Leaf}\;\tA \mid \hsc{Node}\;(\hsc{List}\;(\hsc{Indirect}\;\tA)) \\
  & \kwd{data}\;\hsc{Neg} =
  \hsc{App}\;\hsc{Neg}\;\hsc{Neg} \mid \hsc{Lam}\;(\hsc{Neg} \to \hsc{Neg})
\end{align*}
For formalizing Haskell record definitions, it also conveniently supports selector functions for constructor arguments.

\subsection{Notation}

We avoid using Isabelle notation as much as possible, favoring a Haskell-style syntax for datatype and function definitions. We also use Haskell-style notation $(\dots) \To \dots$ for class constraints on type variables. Isabelle's syntax follows Standard ML in writing type constructors postfix; however, for consistency we use Haskell-style prefix type application throughout. Isabelle type constructors with multiple arguments are shown as tupled.

We consistently use Greek letters $\tA$, $\tB$, $\tC$ for type variables of kind $*$, and $\tT$ for kind $* \to *$. Latin letters $a, b, c$ are program variables, with $f, g, h, k$ referring specifically to functions. We often use sub- and superscripts to annotate polymorphic functions with their types; e.g., $\fmap^\tT_{\tA,\,\tB}$ means $\fmap :: (\tA \to \tB) \to \tT\;\tA \to \tT\;\tB$.

%%%%%%%%%%%%%%%%%%%%%%%%%%%%%%%%%%%%%%%%%%%%%%%%%%
\section{Deflation model of types}
\label{sec:deflation-model}

HOLCF provides a special domain $\D$ whose values are \emph{deflations}, a certain kind of idempotent functions. Deflations are used to model types: To each ``representable'' domain type in HOLCF, we associate a representation of type $\D$. The primary reason for having this model in HOLCF is to implement the Domain package: The deflation model of types lets us reason about the existence of solutions to domain equations, because we can construct recursively defined deflations to represent them.

The Tycon library takes advantage of this existing model of types to derive further benefits. The deflation model gives us a way to express the relationship between different type instances of polymorphic functions, letting us reason about polymorphism. It also lets us reason about type quantification, by quantifying over deflations.

In the remainder of this section, we describe the underlying concepts behind the deflation model, as well as its implementation in HOLCF.

\subsection{Embedding-projection pairs and deflations}

Some cpos can be embedded within other cpos. The concept of an \emph{embedding-projection pair} (often shortened to \emph{ep-pair}) formalizes this notion.
\begin{definition}
Continuous functions $e :: \tA\to\tB$ and $p :: \tB\to\tA$ form an \emph{embedding-projection pair} (or \emph{ep-pair}) if $p \circ e = \hsc{id}_\tA$ and $e \circ p \below \hsid_\tB$. In this case, we write $(e, p) : \tA \eppair \tB$.
\end{definition}

Ep-pairs have many useful properties: $e$ is injective, $p$ is surjective, both are strict, each function uniquely determines the other, and the image of $e$ is a sub-cpo of $\tB$. The composition of two ep-pairs yields another ep-pair: If $(e_1, p_1) : \tA \eppair \tB$ and $(e_2, p_2) : \tB \eppair \tC$, then $(e_2 \circ e_1, p_1 \circ p_2) : \tA \eppair \tC$. Ep-pairs can also be lifted over many type constructors, including strict sums, products, and continuous function space.

\begin{definition}
A continuous function $d :: \tA \to \tA$ is a \emph{deflation} if it is idempotent and below the identity function: $d \circ d = d \below \hsid_\tA$.
\end{definition}

Deflations and ep-pairs are closely related. Given an ep-pair $(e, p) : \tA \eppair \tB$, the composition $e \circ p$ is a deflation on $\tB$ whose image is isomorphic to $\tA$. Conversely, every deflation $d :: \tB \rightarrow \tB$ also gives rise to an ep-pair. Let the cpo $\tA$ be the image of $d$; also let $e$ be the inclusion map from $\tA$ to $\tB$, and let $p = d$. Then $(e, p)$ is an embedding-projection pair. So saying that there exists an ep-pair from $\tA$ to $\tB$ is equivalent to saying that there exists a deflation on $\tB$ whose image is isomorphic to $\tA$. Figure~\ref{fig:ep-pairs} shows the relationship between ep-pairs and deflations.

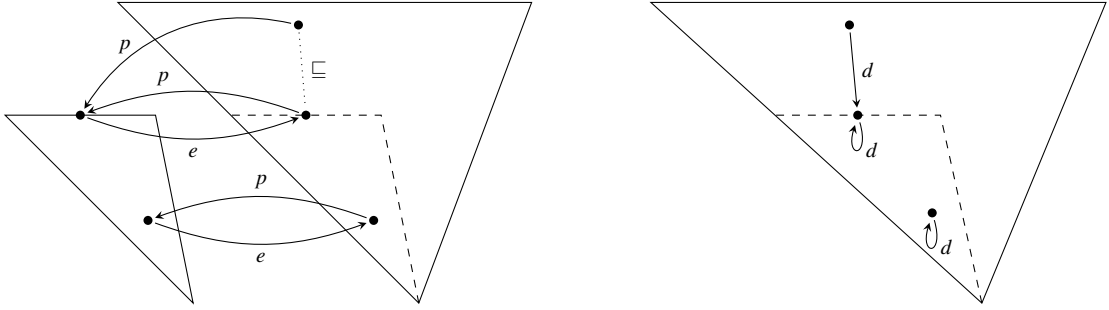
\begin{figure*}
\centering
\hfill
%\subfloat[ep-pair]{
\begin{tikzpicture}
[>=stealth, point/.style={coordinate, circle, fill, inner sep=0, outer sep=0.4mm, minimum size=1.2mm}, xscale=1.0]
\draw [xshift=-3cm] (0,0) -- (-2.5, 2.5) -- (-0.5, 2.5) -- cycle;
\draw (0, 0) -- (-4, 4) -- (1.5, 4) -- cycle;
\draw [dashed] (0, 0) -- (-0.5, 2.5) -- (-2.5, 2.5);
\node (1a) [point] at (-0.6, 1.1) {};
\node (1b) [point, xshift=-3cm] at (-0.6, 1.1) {};
\node (2a) [point] at (-1.5, 2.5) {};
\node (2b) [point, xshift=-3cm] at (-1.5, 2.5) {};
\node (3a) [point] at (-1.6, 3.7) {};
\path (1a) edge [->, bend right=20, above] node {\small $p$} (1b);
\path (1b) edge [->, bend right=20, below] node {\small $e$} (1a);
\path (2a) edge [->, bend right=20, above, pos=0.65] node {\small $p$} (2b);
\path (2b) edge [->, bend right=20, below] node {\small $e$} (2a);
\path (3a) edge [->, bend right=35, above, near end] node {\small $p$} (2b);
\path (3a) edge [dotted, right] node {\small $\sqsubseteq$} (2a);
\end{tikzpicture}
%}
\hfill
%\subfloat[deflation]{
\begin{tikzpicture}
[>=stealth, point/.style={coordinate, circle, fill, inner sep=0, outer sep=0.4mm, minimum size=1.2mm}, xscale=1.1]
\draw (0, 0) -- (-4, 4) -- (1.5, 4) -- cycle;
\draw [dashed] (0, 0) -- (-0.5, 2.5) -- (-2.5, 2.5);
\node (1a) [point] at (-0.6, 1.2) {};
\node (2a) [point] at (-1.5, 2.5) {};
\node (3a) [point] at (-1.6, 3.7) {};
\path (1a) edge [->, loop below, right] node {\small $d$} (1a);
\path (2a) edge [->, loop below, right] node {\small $d$} (2a);
\path (3a) edge [->, right] node {\small $d$} (2a);
\end{tikzpicture}
%}
\hfill{}
\caption{Embedding-projection pairs and deflations}
\label{fig:ep-pairs}
\end{figure*}

A deflation is a function, but it can also be viewed as a set: Just take the image of the function, or equivalently, its set of fixed points---for idempotent functions they are the same. The dashed outline in Fig.~\ref{fig:ep-pairs} shows the set defined by the deflation $d$. Every deflation on a cpo $\tA$ gives a set that is a sub-cpo, and contains $\bot$ if $\tA$ has a least element. Not all sub-cpos have a corresponding deflation, but if one exists then it is unique. The set-oriented and function-oriented views of deflations also give the same ordering: For any deflations $f$ and $g$, $f \sqsubseteq g$ if and only if $\mathrm{Im}(f) \subseteq \mathrm{Im}(g)$.

\subsection{Representable types}
\label{sec:universal-deflation}

We say that a type $\tA$ is \emph{representable} in domain $\U$ if there exists an ep-pair from $\tA$ to $\U$, or equivalently if there exists a deflation $d$ on $\U$ whose image $\mathrm{Im}(d)$ is isomorphic to $\tA$. We say that $\U$ is a \emph{universal domain} for some class of cpos if every cpo in the class is representable in $\U$. Isabelle/HOLCF provides such a universal domain type $\U$, which can represent any bifinite domain---this is a large class of cpos that includes (but is not limited to) all Haskell datatypes \cite{Huffman2009}.

HOLCF defines an axiomatic class of representable domains. The class fixes operations $\hsemb$ and $\hsproj$, and assumes that they form an ep-pair into the universal domain.
\begin{align*}
  & \kwd{class}\:\hsc{Rep}\:\tA\:\kwd{where} \\[-\jot]
  & \hspace{8pt}
  \begin{alignedat}{2}
    & \hsemb &&:: \tA \to \U \\[-\jot]
    & \hsproj &&:: \U \to \tA
  \end{alignedat} \\
  & \hspace{8pt}
  \begin{alignedat}{2}
    & \hsproj \circ \hsemb = \hsid_\tA \\[-\jot]
    & \hsemb \circ \hsproj \below \hsid_\U
  \end{alignedat}
\end{align*}

The universal domain type itself is trivially representable, using identity functions. For other base types like $\hsone$, the HOLCF universal domain library provides appropriate ep-pairs (Fig.~\ref{fig:udom-ep-pairs}).
\begin{align*}
&
\begin{aligned}
  & \kwd{instance}\:\hsRep\:\U\:\kwd{where} \\[-\jot]
  & \hspace{8pt}
  \begin{alignedat}{2}
    & \hsemb &&= \hsid_\U \\[-\jot]
    & \hsproj &&= \hsid_\U
  \end{alignedat}
%\displaybreak[0] \\[\jot]
\end{aligned}
&&
\begin{aligned}
  & \kwd{instance}\:\hsRep\:\hsone\:\kwd{where} \\[-\jot]
  & \hspace{8pt}
  \begin{alignedat}{2}
    & \hsemb &&= \embedding_\hsone \\[-\jot]
    & \hsproj &&= \projection_\hsone
  \end{alignedat}
\end{aligned}
\end{align*}

HOLCF defines the domain $\D$ of deflations over the universal domain as a subtype of $\U\to\hair\U$. (In the Isabelle formalization, explicit conversions between types $\D$ and $\hair\U\to\hair\U$ are always required, but we will keep those implicit here.)
\begin{equation}
\kwd{typedef}\;\D = \{\,d::\U\to\hair\U \mid {d}\circ{d}={d}\below\hsid_\U\,\}
\end{equation}

\begin{definition}[Representation of a type]
Given any representable type $\tA$, we can construct its \emph{representation} (a deflation of type $\D$) by composing \hsc{emb} and \hsc{proj}. We denote the mapping from types to deflations as follows:
\begin{equation*}
\REP{\tA} \defeq \hsemb_\tA \circ \hsproj_\tA
\end{equation*}
\end{definition}

Note that the representation of the universal domain $\REP{\hair\U}$ is therefore the identity deflation, which is maximal among all deflations. Thus we have $\REP{\tA} \below \REP{\hair\U}$ for all representable types $\tA$.

\begin{figure}
\begin{align*}
& \embedding_\to :: (\U\to\hair\U)\to\hair\U &
& \projection_\to :: \U\to(\U\to\hair\U) \\
& \embedding_\otimes :: (\U\otimes\hair\U)\to\hair\U &
& \projection_\otimes :: \U\to(\U\otimes\hair\U) \\
& \embedding_\oplus :: (\U\oplus\hair\U)\to\hair\U &
& \projection_\oplus :: \U\to(\U\oplus\hair\U) \\
& \embedding_\bot :: \U_\bot\to\hair\U &
& \projection_\bot :: \U\to\hair\U_\bot \\
& \embedding_\hsone :: \hsone\to\hair\U &
& \projection_\hsone :: \U\to\hsone
\end{align*}
\caption{Embedding-projection pairs provided by the universal domain library in HOLCF}
\label{fig:udom-ep-pairs}
\end{figure}

\subsection{Representable type constructors}
\label{sec:representable-tycons}

While types can be represented by deflations, type \emph{constructors} (which are like functions from types to types) can be represented as functions from deflations to deflations. We say that a type constructor $F :: *\to*$ is representable in $\U$ if there exists a continuous function $F_\D :: \D \to \D$ such that $\REP{F(\tA)} = F_\D\,\REP{\tA}$. Such deflation combinators can be used to build deflations for recursive datatypes \cite[\S7]{Gunter1990}. This is precisely the technique used by recent versions of the Domain package \cite[\S6.6]{holcf11}; we will also use the same technique for creating new type constructors in Sec.~\ref{sec:instantiation}.

We have a recipe for setting up a primitive HOLCF type as a representable type constructor: We just need an ep-pair to the universal domain and a map function. (The HOLCF universal domain library provides a selection of suitable ep-pairs; see Fig.~\ref{fig:udom-ep-pairs}.) We now demonstrate the recipe using the strict product type.
\begin{align*}
  & \mapProd :: (\tA \to \tA'\!,\:\tB \to \tB') \to \tA \otimes \tB \to \tA'\! \otimes \tB' \\[-\jot]
  & \mapProd\,(f, g)\,(\hsc{SPair}\:x\:y) = \hsc{SPair}\,(f\:x)\,(g\:y)
  \\[\jot]
  & (\otimes_\D) :: \D \to \D \to \D \\[-\jot]
  & a \otimes_\D b = \embedding_\otimes \circ \mapProd\,(a, b) \circ \projection_\otimes
  \\[\jot]
  & \kwd{instance}\:(\hsRep\:\tA, \hsRep\:\tB) \To \hsRep\,(\tA \otimes \tB)\:\kwd{where}
  \\[-\jot]
  & \hspace{8pt}
  \begin{alignedat}{2}
    & \hsemb &&= \embedding_\otimes \circ \mapProd\,(\hsemb_\tA, \hsemb_\tB)
    \\[-\jot]
    & \hsproj &&= \mapProd\,(\hsproj_\tA, \hsproj_\tB) \circ \projection_\otimes
  \end{alignedat}
\end{align*}
The reader can verify that $(\otimes_\D)$ does in fact preserve deflations, that $\hsemb$ and $\hsproj$ do form an ep-pair for type $\tA\otimes\tB$, and that $(\otimes_\D)$ actually does represent the strict product type constructor: $\REP{\tA\otimes\tB} = \REP{\tA} \otimes_\D \REP{\hair\tB}$.

Most other HOLCF type constructors work exactly like the strict product. However, the continuous function space is special because it is contravariant in its first argument.
\begin{align*}
  & \mapFun :: (\tA'\!\to \tA,\:\tB \to \tB') \to (\tA \to \tB) \to (\tA'\! \to \tB') \\[-\jot]
  & \mapFun\,(f, g)\:h = g \circ h \circ f
  \\[\jot]
  & (\to_\D) :: \D \to \D \to \D \\[-\jot]
  & a \to_\D b = \embedding_\to \circ \mapFun\,(a, b) \circ \projection_\to
  \\[\jot]
  & \kwd{instance}\:(\hsRep\:\tA, \hsRep\:\tB) \To \hsRep\,(\tA \to \tB)\:\kwd{where}
  \\[-\jot]
  & \hspace{8pt}
  \begin{alignedat}{2}
    & \hsemb &&= \embedding_\to \circ \mapFun\,(\hsproj_\tA, \hsemb_\tB)
    \\[-\jot]
    & \hsproj &&= \mapFun\,(\hsemb_\tA, \hsproj_\tB) \circ \projection_\to
  \end{alignedat}
\end{align*}
Due to contravariance, the first argument to $\mapFun$ has type $\tA'\!\to\tA$ instead of $\tA\to\tA'$. Also note that in the $\hsRep$ instance, $\hsemb$ calls $\hsprj$ and vice versa. Otherwise everything works similarly to the other types.

\subsection{Coercion}
\label{sec:coercion}

We can write a function to coerce between any two representable types: First embed into the universal domain $\U$, and then project out to a different type.
\begin{align*}
  & \hscoerce :: (\hsRep\:\tA, \hsRep\:\tB) \To \tA \to \tB
  \\[-\jot]
  & \hscoerce_{\tA,\,\tB} = \hsproj_\tB \circ \hsemb_\tA
\end{align*}

Our primary use for coercion will be to relate different type instances of polymorphic functions. In the remainder of the paper, we will often need to prove properties about coerced values; to facilitate this, we assemble a collection of simplification rules. First of all, $\hscoerce$ may reduce to $\hsemb$, $\hsprj$, or $\hsid$, depending on the type:
\begin{align}
\hscoerce_{\tA,\tA} &= \hsid_\tA \\
\hscoerce_{\tA,\U} &= \hsemb_\tA \\
\hscoerce_{\U,\tA} &= \hsproj_\tA
\end{align}

Other properties about $\hscoerce$ depend on the relative ``sizes'' of the source and target types. A coercion from a smaller to a larger type is injective (an embedding, in fact). Coercing twice in a row is the same as coercing once, as long as the intermediate type is larger than one of the source or target types.
\begin{gather}
\inferrule
  {\REP{\tA}\below\REP{\hair\tB} \vee \REP{\tC} \below \REP{\hair\tB}}
  {\hscoerce_{\tB,\tC} \circ \hscoerce_{\tA,\,\tB} = \hscoerce_{\tA,\tC}}
\end{gather}

Coercing between similar datatypes is the same as mapping $\hscoerce$ over the elements. (As an exercise, the reader may wish to verify Eq.~\eqref{eq:coerce-sprod} by expanding the definitions given earlier in Sec.~\ref{sec:representable-tycons}.) Using these rules, it is easy to verify that $\hscoerce$ commutes with each data constructor.
\begin{align}
\label{eq:coerce-sprod}
\hscoerce_{\tA\otimes\tB,\tC\otimes\tD} &= \mapProd(\hscoerce_{\tA,\tC}, \hscoerce_{\tB,\tD}) \\
\hscoerce_{\tA\oplus\tB,\tC\oplus\tD} &= \mapSum(\hscoerce_{\tA,\tC}, \hscoerce_{\tB,\tD}) \\
\hscoerce_{\tA_\bot,\,\tB_\bot} &= \mapLift(\hscoerce_{\tA,\,\tB})
\end{align}

A similar rule holds for coercions between two function types. The expanded form in Eq.~\eqref{eq:coerce-cfun} will be particularly useful for simplifying coercions in later proofs.
\begin{align}
\hscoerce_{(\tA\to\tB),(\tC\to\tD)} &= \mapFun(\hscoerce_{\tC,\tA}, \hscoerce_{\tB,\tD}) \\
\label{eq:coerce-cfun}
\hscoerce_{(\tA\to\tB),(\tC\to\tD)}\,f &= \hscoerce_{\tB,\tD} \circ f \circ \hscoerce_{\tC,\tA}
\end{align}

A note about the ubiquity of the \hsc{Rep} class: For the remainder of this paper, we will assume that all types $\tA, \tB, \tC, \dots$ are in class \hsc{Rep}, without writing \hsc{Rep} class constraints explicitly. The reader may treat \hsc{emb}, \hsc{proj}, and \hsc{coerce} as if they were completely polymorphic. (HOLCF achieves a similar effect using the ``default sort'' mechanism, assigning all type variables to class \hsc{Rep} unless annotated otherwise.)

%%%%%%%%%%%%%%%%%%%%%%%%%%%%%%%%%%%%%%%%%%%%%%%%%%
\section{Type constructor classes in the Tycon library}
\label{sec:constructor-classes}

%[TODO: overview]

%%%%%%%%%%%%%%%%%%%%%%%%%%%%%%%%%%%%%%%%%%%%%%%%%%
\subsection{Class Tycon and type application}
\label{sec:tycon}

In the Haskell type expression $\tT\:\tA$, the two type variables have different kinds: Say $\tA$ is an ordinary type of kind $*$; then $\tT$ may be a type constructor of kind $*\to*$. Isabelle's type system was not designed to be this expressive: All type variables in Isabelle represent ordinary types (corresponding to Haskell kind $*$).

Our solution to this limitation (originally introduced in \cite{HMW2005}) is to define a binary Isabelle type constructor $(-\cdot-)$ that models Haskell type application. The right argument must be in the Isabelle class $\hsRep$, which models Haskell kind $*$. The left argument must be in a new class \hsc{Tycon}, which models Haskell kind $*\to*$.

Class \hsc{Tycon} is defined as follows. It has no axioms, but fixes a single constant which is a deflation constructor.\footnote{In the Isabelle formalization, we express the dependence of $\TC{\tT}$ on type $\tT$ by adding a dummy function argument whose type is a phantom type mentioning $\tT$.}
\begin{equation*}
\kwd{class}\:\hsc{Tycon}\:\tT\:\kwd{where}\:\TC{\tT} :: \D \to \D
\end{equation*}

Now we want to define type $\tT\cdot\tA$ so that $\REP{\tT\cdot\tA} = \TC{\tT}\REP{\tA}$. We therefore define $\tT\cdot\tA$ as a subtype of $\U$, consisting of the image (or equivalently, the set of fixed-points) of the deflation $\TC{\tT}\REP{\tA}$.
\begin{align*}
  & \kwd{typedef}\:(\hsc{Tycon}\:\tT, \hsc{Rep}\:\tA) \To \tT\cdot\tA
  \\[-\jot]
  & \hspace{8pt}
  = \{\,u::\U \mid \TC{\tT}\REP{\tA}\,{u}={u}\,\}
  \\[\jot]
  & \kwd{instance}\:(\hsc{Tycon}\:\tT, \hsc{Rep}\:\tA) \To \hsc{Rep}\,(\tT\cdot\tA)\:
  \kwd{where}
  \\[-\jot]
  & \hspace{8pt}
  \begin{alignedat}{2}
    & \hsemb\:x = x \\[-\jot]
    & \hsprj\:u = \TC{\tT}\REP{\tA}\,u
  \end{alignedat}
\end{align*}
Note that the definitions of $\hsemb$ and $\hsproj$ contain implicit coercions between $\U$ and $\tT\cdot\tA$. The desired representation property then follows directly from the definitions of $\hsemb$ and $\hsproj$:
\begin{equation}
\REP{\tT\cdot\tA} = \TC{\tT}\REP{\tA}
\end{equation}

It is worth pointing out that while the construction refers to the deflation combinator $\TC{\tT}$, actual values of type $\tT$ are never used anywhere. This is consistent with Haskell, where there are no values inhabiting higher-kinded types.

%%%%%%%%%%%%%%%%%%%%%%%%%%%%%%%%%%%%%%%%%%%%%%%%%%
\subsection{Class Functor}

The Haskell \hsc{Functor} class is for types that can be mapped over.
%
%From GHC documentation: ``The Functor class is used for types that can be mapped over. Instances of Functor should satisfy the following laws \dots The instances of Functor for lists, Maybe and IO satisfy these laws.''
%
\begin{align*}
& \kwd{class}\:\hsc{Functor}\:\tT\:\kwd{where} \\[-\jot]
& \hspace{8pt} \fmap :: (\tA \to \tB) \to (\tT\:\tA \to \tT\:\tB)
\end{align*}
Each Haskell \hsc{Functor} instance should satisfy the identity and composition laws:
\begin{align}
  \fmap\:\hsid &= \hsid \\
  \fmap\:(f \circ g) &= \fmap\:f \circ \fmap\:g
\end{align}

How close are we to being able to formalize this in Isabelle? Using the type application machinery from Sec.~\ref{sec:tycon}, we can at least express the class constraint $\hsc{Functor}\:\tT$ and the result type of $\fmap$. However, there are still some problems. First, let us examine the type of $\fmap$ more closely. The type constructor variable $\tT$ is fixed, but types $\tA$ and $\tB$ are actually universally quantified:
\begin{equation*}
\fmap^\tT :: \forall\,\tA\:\tB.\,(\tA \to \tB) \to (\tT\cdot\tA \to \tT\cdot\tB)
\end{equation*}
The problem is that Isabelle's class system does not allow polymorphic class constants. Isabelle's type system does not support first-class polymorphism, and the type of class functions are only allowed to contain one free type variable, i.e., the one mentioned in the class constraint \cite{Wenzel1997}.

The solution is to move the polymorphism out of the class declaration. We replace the polymorphic $\fmap^\tT$ with a single, monomorphic constant representing $\fmap^\tT_{\U,\U}$, the ``largest'' type instance of $\fmap^\tT$. (We use the underlined name $\fmapU^\tT$ to refer to this monomorphic version.) We then define the polymorphic $\fmap^\tT$ by coercion from $\fmapU^\tT$.
\begin{align*}
  & \kwd{class}\:(\hsc{Tycon}\:\tT) \To \hsc{Functor}\:\tT\:\kwd{where}
  \\[-\jot]
  & \hspace{8pt} \fmapU^\tT :: (\U \to \U) \to (\tT\cdot\U \to \tT\cdot\U)
  \\[\jot]
  & \fmap :: (\hsc{Functor}\:\tT) \To (\tA \to \tB) \to \tT\cdot\tA \to \tT\cdot\tB
  \\[-\jot]
  & \fmap^\tT_{\tA,\,\tB} = \hscoerce\:\fmapU^\tT
\end{align*}

In Haskell, polymorphically typed functions like $\fmap^\tT$ are always \emph{parametrically} polymorphic. That is, parametricity (a meta-property of the type system) ensures that all of the different type instances of $\fmap^\tT_{\tA,\,\tB}$ behave uniformly \cite{Wadler1989}. Isabelle's type system does not provide any automatic parametricity guarantees, but by defining all type instances of $\fmap^\tT$ by coercion from a single constant, we ensure a similar kind of uniformity across type instances in the our library.

The formalization of class \hsc{Functor} is yet incomplete: We have the constant, but not the functor laws. We need to find a set of class axioms about $\fmapU^\tT$ that will let us derive the polymorphic functor laws about $\fmap^\tT$.

As a first try, we might just write down the functor laws with all the types specialized to type $\U$:
\begin{align}
\label{eq:fmapu-id}
\fmapU^\tT\:\hsid_\U &= \hsid_{\tT\cdot\U} \\
\label{eq:fmapu-compose}
\fmapU^\tT (f \circ g) &= \fmapU^\tT f \circ \fmapU^\tT g
\end{align}
However, we shall treat these as tentative until we see whether they are sufficient to derive the polymorphic functor laws.

\begin{theorem}[Identity]
\label{thm:functor-identity}
For any $\tT$ in class \hsc{Functor} and representable type $\tA$, the functor identity law holds:
\begin{equation*}
\fmap^\tT_{\tA,\tA}\:\hsid_\tA = \hsid_{\tT\cdot\tA}
\end{equation*}
\end{theorem}

\begin{proof}
We start by unfolding the definition of $\fmap$ and rewriting with properties of $\hscoerce$.
\begin{align*}
& \fmap^\tT_{\tA,\tA}\:\hsid_\tA \\
&= (\hscoerce\:\fmapU^\tT)\:\hsid_\tA \\
&= \hscoerce_{(\tT\cdot\hair\U,\tT\cdot\tA)} \circ \fmapU^\tT (coerce\:\hsid_\tA) \circ \hscoerce_{(\tT\cdot\tA,\tT\cdot\hair\U)} \\
&= \hscoerce_{(\tT\cdot\hair\U,\tT\cdot\tA)} \circ \fmapU^\tT \REP{\tA} \circ \hscoerce_{(\tT\cdot\tA,\tT\cdot\hair\U)}
\end{align*}

\noindent
At this point we are stuck. A law about $\fmapU^\tT\:\hsid_\U$ does not help here, because coercing $\hsid_\tA$ to type $\U\to\U$ does not yield $\hsid_\U$; it gives $\REP{\tA}$ instead. What we really need is a rewrite rule for $\fmapU^\tT$ applied to an arbitrary deflation. The class axiom must assert that the map function $\fmapU^\tT$ ``agrees'' with the deflation combinator $\TC{\tT}$ in a certain sense.

\begin{definition}
We say that a function $f$ on a representable type $\tA$ \emph{agrees with} a deflation $d$ on the universal domain, if $f$ coerced to type $\U\to\U$ is equal to $d$ (regarded as a function).
\begin{equation*}
(f::\tA\to\tA) \isodefl (d::\D) \overset{\mathrm{def}}{\Longleftrightarrow} \hsc{emb}_\tA \circ f \circ \hsc{proj}_\tA = d
\end{equation*}
\end{definition}

\noindent
This agreement relation is already present in HOLCF: It is used internally by the Domain package for relating deflation combinators to map functions, for proving identity laws and deriving induction rules \cite{holcf11}. So it is fitting that it should appear in this situation, where we are again proving functor identity laws.

We replace Eq.~\eqref{eq:fmapu-id} with this generalized class axiom, shown here also in its unfolded form:
\begin{gather}
\fmapU^\tT\,(d::\D) \isodefl \TC{\tT}\,d \\
\hsemb_{\tT\cdot\U} \circ \fmapU^\tT\,(d::\D) \circ \hsprj_{\tT\cdot\U} = \TC{\tT}\,d
\end{gather}

\noindent
Now we can continue where we left off:
\begin{align*}
& \hscoerce_{(\tT\cdot\hair\U,\tT\cdot\tA)} \circ \fmapU^\tT \REP{\tA} \circ \hscoerce_{(\tT\cdot\tA,\tT\cdot\hair\U)} \\
& = \hsc{proj}_{\tT\cdot\tA} \circ \hsc{emb}_{\tT\cdot\hair\U} \circ \fmapU^\tT \REP{\tA} \circ \hsc{proj}_{\tT\cdot\hair\U} \circ \hsc{emb}_{\tT\cdot\tA} \\
&= \hsc{proj}_{\tT\cdot\tA} \circ \TC{\tT}\REP{\tA} \circ \hsc{emb}_{\tT\cdot\tA} \\
&= \hsc{proj}_{\tT\cdot\tA} \circ \REP{\tT\cdot\tA} \circ \hsc{emb}_{\tT\cdot\tA} \\
&= \hsc{proj}_{\tT\cdot\tA} \circ \hsc{emb}_{\tT\cdot\tA} \circ \hsc{proj}_{\tT\cdot\tA} \circ \hsc{emb}_{\tT\cdot\tA} \\
&= \hsid_{\tT\cdot\tA} \circ \hsid_{\tT\cdot\tA} \\
&= \hsid_{\tT\cdot\tA}\qedhere
\end{align*}
\end{proof}

\begin{theorem}[Composition]
\label{thm:functor-composition}
For any $\tT$ in class \hsc{Functor} and functions $f::\tB\to\tC$ and $g::\tA\to\tB$, the functor composition law holds:
\begin{equation*}
\fmap^\tT_{\tA,\tC}\,(f \circ g) = \fmap^\tT_{\tB,\tC}\:f \circ \fmap^\tT_{\tA,\,\tB}\:g
\end{equation*}
\end{theorem}

\begin{proof}
We rewrite both sides of the equation, trying to reduce it to a trivial equality. We start by unfolding the definition of $\fmap$.
\begin{align*}
\fmap^\tT_{\tA,\tC}\,(f \circ g) &= \fmap^\tT_{\tB,\tC}\:f \circ \fmap^\tT_{\tA,\,\tB}\:g \\
(\hscoerce\:\fmapU^\tT)(f \circ g) &= (\hscoerce\:\fmapU^\tT)\,f \circ (\hscoerce\:\fmapU^\tT)\,g
\end{align*}

\noindent
After rewriting using Eq.~\eqref{eq:coerce-cfun}, we have
\begin{multline*}
\hscoerce \circ \fmapU^\tT(\hscoerce\,(f \circ g)) \circ \hscoerce \\[-\jot]
= \hscoerce \circ \fmapU^\tT(\hscoerce\:f) \circ \hscoerce \\[-\jot]
\circ \hscoerce \circ \fmapU^\tT(\hscoerce\:g) \circ \hscoerce
\end{multline*}

In the middle of the right-hand side we have two adjacent coercions, where we go from type $\tT\cdot\U$ to $\tT\cdot\tB$ and back. Since the intermediate type $\tT\cdot\tB$ is smaller, they do not cancel completely. It turns out that they reduce to an application of $\fmapU^\tT$:
\begin{align}
\notag & \hscoerce_{\tT\cdot\tB,\tT\cdot\U} \circ \hscoerce_{\tT\cdot\U,\tT\cdot\tB} \\
\notag & = \hsprj_{\tT\cdot\U} \circ \hsemb_{\tT\cdot\tB} \circ \hsprj_{\tT\cdot\tB} \circ \hsemb_{\tT\cdot\U} \\
\notag & = \hsprj_{\tT\cdot\U} \circ \REP{\tT\cdot\tB} \circ \hsemb_{\tT\cdot\U} \\
\notag & = \hsprj_{\tT\cdot\U} \circ \TC{\tT}\REP{\hair\tB} \circ \hsemb_{\tT\cdot\U} \\
\notag & = \hsprj_{\tT\cdot\U} \circ \hsemb_{\tT\cdot_\U} \circ \fmapU^\tT\REP{\hair\tB} \circ \hsprj_{\tT\cdot\U} \circ \hsemb_{\tT\cdot\U} \\
\notag & = \hsid_{\tT\cdot\U} \circ \fmapU^\tT\REP{\hair\tB} \circ \hsid_{\tT\cdot\U} \\
\label{eq:coerce-fmap} & = \fmapU^\tT\REP{\hair\tB}
\end{align}

Rewriting the right-hand side with Eq.~\eqref{eq:coerce-fmap} yields three occurrences of $\fmapU^\tT$ composed together. Using the composition rule \eqref{eq:fmapu-compose} to collapse these, we get
\begin{multline*}
\hscoerce \circ \fmapU^\tT(\hscoerce\,(f \circ g)) \circ \hscoerce \\[-\jot]
= \hscoerce \circ \fmapU^\tT(\hscoerce\:f \circ \REP{\hair\tB} \circ \hscoerce\:g) \circ \hscoerce
\end{multline*}

Finally, it only remains to show that the arguments to $\fmapU^\tT$ on each side are equal. We work from right to left.
\begin{align*}
& \hscoerce_{(\tB\to\tC),(\U\to\U)}\,f \circ \REP{\hair\tB} \circ \hscoerce_{(\tA\to\tB),(\U\to\U)}\,g \\
& = \hsemb_\tC \circ f \circ \hsprj_\tB \circ \REP{\hair\tB} \circ \hsemb_\tB \circ g \circ \hsprj_\tA \\
& = \hsemb_\tC \circ f \circ \hsprj_\tB \circ \hsemb_\tB \circ \hsprj_\tB \circ \hsemb_\tB \circ g \circ \hsprj_\tA \\
& = \hsemb_\tC \circ f \circ \hsid_\tB \circ \hsid_\tB \circ g \circ \hsprj_\tA \\
& = \hsemb_\tC \circ f \circ g \circ \hsprj_\tA \\
& = \hscoerce_{(\tA\to\tC),(\U\to\U)}\,(f \circ g) \qedhere
\end{align*}
\end{proof}

The final formulation of class \hsc{Functor}, complete with the generalized identity law, is shown in Fig.~\ref{fig:functor}.

We should note that while transfer proofs like Theorem~\ref{thm:functor-composition} may look lengthy on paper, they are actually highly automated in Isabelle: Most such proofs require only a single call to Isabelle's simplifier, as long as the appropriate extra rewrite rules like Eq.~\eqref{eq:coerce-fmap} are in place. This is important for usability of the library, because users will need to perform similar transfer proofs often---not just when defining new constructor classes, but also when instantiating them.

We present proofs here in a point-free style, with liberal use of the function composition operator $(\circ)$, because it makes the proofs easier to read. However, Isabelle is not so good at reasoning modulo the associativity of function composition. Automatic proofs by rewriting work better with nested function applications, e.g. $f(g(h(x)))$ rather than $f \circ g \circ h$. Therefore, the rewrite rules and other theorems in our library are actually formalized using fully applied functions instead of function composition.

\begin{figure}
\begin{align*}
  & \kwd{class}\:\hsc({Tycon}\:\tT) \To \hsc{Functor}\:\tT\:\kwd{where} \\[-\jot]
  & \hspace{8pt} \fmapU^\tT :: (\U \to \U) \to (\tT\cdot\U \to \tT\cdot\U) \\[\jot]
  & \hspace{8pt}
  \begin{alignedat}{2}
    & \fmapU^\tT\,(d::\D) &&\isodefl \TC{\tT}\,d \\[-\jot]
    & \fmapU^\tT (f \circ g) &&= \fmapU^\tT f \circ \fmapU^\tT g
  \end{alignedat}
  \\[\jot]
  & \fmap :: (\hsc{Functor}\:\tT) \To (\tA \to \tB) \to \tT\cdot\tA \to \tT\cdot\tB
  \\[-\jot]
  & \fmap^\tT_{\tA,\tB} = \hscoerce\:\fmapU^\tT
\end{align*}
\caption{Isabelle \hsc{Functor} class}
\label{fig:functor}
\end{figure}

\subsection{Generic theorems about functors}

Now that we have a functor class, we can prove further theorems about $\fmap$. Here is an example theorem, about its strictness. The proof uses only the functor laws and basic properties of domain theory; the result is applicable to any valid functor instance.

\begin{theorem}[Strict $\fmap$]
\label{thm:strict-fmap}
If $f::\tA\to\tB$ is a strict function, then $\fmap\:f$ is also strict: $f\:\bot = \bot \implies \fmap\:f\:\bot = \bot$.
\end{theorem}
\begin{proof}
Fix $f :: \tA\to\tB$, and assume $f\:\bot_\tA = \bot_\tB$. Let $g :: \tB\to\tA$ be the constant bottom function, $g\:x = \bot_\tA$. From the strict\-ness of $f$, it follows that $f \circ g = \hsc{const}\:\bot \below id_\tB$. We can now show the goal by antisymmetry and transitivity reasoning:
\begin{align*}
\fmap\:f\:\bot
&\below \fmap\:f\:(\fmap\:g\:\bot) & \justification{monotonicity with $\bot$} \\
&= \fmap\:(f \circ g)\:\bot & \justification{composition law} \\
&\below \fmap\:\hsid\:\bot & \justification{monotonicity, $f \circ g \below \hsid$} \\
&= \bot & \justification{identity law}
\end{align*}
Thus we have $\fmap\:f\:\bot \below \bot$, which implies $\fmap\:f\:\bot = \bot$.
\end{proof}

%%%%%%%%%%%%%%%%%%%%%%%%%%%%%%%%%%%%%%%%%%%%%%%%%%
\subsection{Subclasses of Functor}
\label{sec:functorplus}

Users of the Tycon library can easily formalize additional constructor classes that are subclasses of \hsc{Functor}. The library already contains several examples, and they all follow the same general process.

A constructor class may fix some number of polymorphic constants, and assume a set of polymorphic class axioms. The formalized constructor class fixes a monomorphic version of each polymorphic function, with type variables instantiated to $\U$. Similarly, the formalized class assumes a monomorphic version of each class axiom. The polymorphic version of each functions is defined separately, using coercion. In general, we will also add a \emph{naturality} law for each polymorphic function, which is related to the parametricity property, or free theorem, derived from its type \cite{Wadler1989}. The naturality laws are necessary for transferring properties about the monomorphic constants to the polymorphic ones.

The \hsc{Functor} class is a special case: No extra naturality law was needed for $\fmap$, because the functor composition law \emph{is} the naturality law for $\fmap$. The \hsc{Monad} class is perhaps the primary motivation for this work, but the interactions and redundancies between its laws also make it a bit of a special case. The general principles are best illustrated with a more regular example. So here we present a class \hsc{FunctorPlus}, which fixes a binary append operation for combining functor values:
\begin{align*}
& \kwd{class}\:(\hsc{Functor}\:\tT) \To \hsc{FunctorPlus}\:\tT\:\kwd{where} \\[-\jot]
& \hspace{8pt} (\hsapp)::\tT\:\tA\to\tT\:\tA\to\tT\:\tA
\end{align*}
Each instance of \hsc{FunctorPlus} should also ensure that $(\hsapp)$ is associative:
\begin{equation}
\label{eq:append-assoc}
({x}\hsapp{y})\hsapp{z}={x}\hsapp({y}\hsapp{z})
\end{equation}

Any implementation of $(\hsapp)$ should also satisfy a naturality condition, which essentially states that it commutes with $\fmap$. The form of this law is derived from the polymorphic type of $(\hsapp)$; it holds in Haskell as a consequence of parametricity.
\begin{equation}
\label{eq:fmap-append}
\fmap\:f\:(x \hsapp y) = (\fmap\:f\:x) \hsapp (\fmap\:f\:y)
\end{equation}

We formalize class \hsc{FunctorPlus} in Isabelle according to the general pattern outlined above; the code is shown in Fig.~\ref{fig:functorplus}.

\begin{figure}
\begin{align*}
  & \kwd{class}\:(\hsc{Functor}\:\tT) \To \hsc{FunctorPlus}\:\tT\:\kwd{where}
  \\[-\jot]
  & \hspace{8pt} (\hsappU^\tT) :: \tT\cdot\U \to \tT\cdot\U \to \tT\cdot\U \\
  & \hspace{8pt}
  \begin{alignedat}{2}
    & \fmapU^\tT\,f\:(x \hsappU^\tT y) &&= (\fmapU^\tT\,f\:x) \hsappU^\tT (\fmapU^\tT\,f\:y)
    \\[-\jot]
    & (x \hsappU^\tT y) \hsappU^\tT z &&= x \hsappU^\tT (y \hsappU^\tT z)
  \end{alignedat}
  \\[\medskipamount]
  & (\hsapp) :: (\hsc{FunctorPlus}\:\tT) \To \tT\cdot\tA \to \tT\cdot\tA \to \tT\cdot\tA
  \\[-\jot]
  & (\hsapp^\tT_\tA) = \hscoerce\:(\hsappU^\tT)
\end{align*}
\caption{Isabelle \hsc{FunctorPlus} class}
\label{fig:functorplus}
\end{figure}

The need for the naturality law becomes apparent when transferring laws to the polymorphic version of $(\hsapp)$. When we transfer the associativity law, we get a situation similar to what we had with the proof of Theorem~\ref{thm:functor-composition}: Between the two occurrences of $(\hsapp)$, we get a pair of coercions from $\tT\cdot\U$ to $\tT\cdot\tA$ and back; these reduce to $\fmapU^\tT\,\REP{\tA}$. The naturality law lets us push the $\fmapU^\tT$ into the arguments of the inner append, bringing the two appends together so that the monomorphic associativity rule can be applied. In the end, we are able to prove the polymorphic version of the associativity law with one call to the simplifier. Similarly, we can also derive the polymorphic version of the naturality law in one step.

%%%%%%%%%%%%%%%%%%%%%%%%%%%%%%%%%%%%%%%%%%%%%%%%%%
\subsection{Class Monad}
\label{sec:monad}

The definition of the \hsc{Monad} class should be familiar to every Haskell programmer.
\begin{align*}
  &\kwd{class}\:\hsc{Monad}\:\tT\:\kwd{where} \\[-\jot]
  &\hspace{8pt}
  \begin{alignedat}{2}
    & \hsc{return}&&::\tA\to\tT\:\tA \\[-\jot]
    & (\hsbind)&&::\tT\:\tA\to(\tA\to\tT\:\tB)\to\tT\:\tB
  \end{alignedat}
\end{align*}

\noindent
The standard monad laws are left unit, right unit, and associativity.
\begin{align}
\label{eq:left-unit}
\hsc{return}\:{a}\hsbind{k} &= {k}\:{a} \\
\label{eq:right-unit}
{m}\hsbind\hsc{return} &= {m} \\
\label{eq:bind-bind}
({m}\hsbind{h})\hsbind{k}&={m}\hsbind(\lambda{x}.\:{h}\:{x}\hsbind{k})
\end{align}

To translate this Haskell class definition into Isabelle, we can follow the standard process established in Sec.~\ref{sec:functorplus}: Replace the polymorphic operations with monomorphic ones, where each type variable is instantiated to $\U$; specialize the types in the class axioms to $\U$; and add naturality laws for each of the constants.

Below are the naturality laws for the monad operations, derived from their type signatures. Note that $(\hsbind)$ has two naturality laws, because its type has two polymorphic variables.
\begin{align}
\label{eq:fmap-return}
\fmap\:f\,(\hsc{return}\:a) &= \hsc{return}\,(f\:a) \\
\label{eq:bind-fmap}
(\fmap\:f\:m) \hsbind k &= m \hsbind (k \circ f) \\
\label{eq:fmap-bind}
\fmap\:f\,(m \hsbind k) &= m \hsbind (\fmap\:f \circ k)
\end{align}

Three monad laws plus three naturality laws would make six class axioms in total. However, it is possible to reduce this number. Using Eqs. \eqref{eq:right-unit} and \eqref{eq:bind-fmap}, we can derive a simple definition of $\fmap$ in terms of $(\hsbind)$ and \hsc{return}:
\begin{equation}
\label{eq:monad-fmap}
\fmap\:f\:m = m \hsbind (\hsc{return} \circ f)
\end{equation}
This definition of $\fmap$ is often referred to as a fourth monad law; it is expected to hold for any Haskell type that is an instance of both the \hsc{Functor} and \hsc{Monad} classes.

It is simple to verify that from Eqs. \eqref{eq:left-unit}, \eqref{eq:bind-bind}, and \eqref{eq:monad-fmap}, we can derive all of the other monad and naturality laws. Thus we can use these three to formalize our \hsc{Monad} class (see Fig.~\ref{fig:class-monad}).
\begin{figure}
\begin{align*}
  & \kwd{class}\:(\hsc{Functor}\:\tT)\To\hsc{Monad}\:\tT\:\kwd{where}
  \\[-\jot]
  & \hspace{8pt}
  \begin{alignedat}{2}
    & \univ{\hsc{return}}^\tT&&::\U\to\tT\cdot\U \\[-\jot]
    & (\hsbindU^\tT)&&::\tT\cdot\U\to(\U\to\tT\cdot\U)\to\tT\cdot\U
  \end{alignedat} \\
  & \hspace{8pt}
  \begin{alignedat}{3}
    & \univ{\hsc{return}}^\tT\:{u}\hsbindU^\tT{k} &&= {k}\:{u} \\[-\jot]
    & ({m} \hsbindU^\tT {h})\hsbind{k}
    &&= {m} \hsbindU^\tT (\lambda{x}.\:{h}\:{x} \hsbindU^\tT {k}) \\[-\jot]
    & \fmapU^\tT\,f\:m &&= m \hsbindU^\tT (\univ{\hsc{return}}^\tT \circ f)
  \end{alignedat}
  \\[\medskipamount]
  & \hsc{return} :: (\hsc{Monad}\:\tT) \To \tA\to\tT\cdot\tA \\[-\jot]
  & \hsc{return}^\tT_\tA = \hscoerce\:\univ{\hsc{return}}^\tT \\
  & (\hsbind) :: (\hsc{Monad}\:\tT) \To \tT\cdot\tA \to (\tA \to \tT\cdot\tB) \to \tT\cdot\tB
  \\[-\jot]
  & (\hsbind^\tT_{\tA,\tB}) = \hscoerce\:(\hsbindU^\tT)
\end{align*}
\caption{Isabelle \hsc{Monad} class}
\label{fig:class-monad}
\end{figure}

From the class axioms, we re-derive the rest of the original six laws for the monomorphic constants. Then we transfer all of the laws to the polymorphic constants, using the automated method described previously in Sec.~\ref{sec:functorplus}.

\subsection{Generic theorems about monads}

Using the polymorphic monad laws, we can proceed to prove further theorems about arbitrary monads---for example, a property about the strictness of the bind operator.

\begin{theorem}[Strict $\hsbind$]
\label{thm:strict-bind}
Bind is strict in its first argument, if its second argument is also strict: $k\:\bot = \bot \Longrightarrow \bot \hsbind k = \bot$.
\end{theorem}
\begin{proof}
By antisymmetry, it suffices to show $\bot \hsbind k \below \bot$.
\begin{align*}
  \bot \hsbind k
  &\below \hsc{return}\:\bot \hsbind k & \justification{monotonicity} \\
  &= k\:\bot & \justification{left unit law} \\
  &= \bot & \justification{strictness of $k$} & \qedhere
\end{align*}
\end{proof}

Within the context of class \hsc{Monad}, we can also define derived monadic constants, such as \hsc{join}.
\begin{align*}
  & \hsc{join} :: (\hsc{Monad}\;\tT) \To \tT\cdot(\tT\cdot\tA) \to \tT\cdot\tA
  \\[-\jot]
  & \hsc{join}\;m = m \hsbind \hsid
\end{align*}
We can derive a collection of standard lemmas about \hsc{join} by unfolding its definition and rewriting with the monad laws. These lemmas will then be valid for any type in the \hsc{Monad} class.

%%%%%%%%%%%%%%%%%%%%%%%%%%%%%%%%%%%%%%%%%%%%%%%%%%
\section{Instantiating type constructor classes}
\label{sec:instantiation}

Type constructor classes like \hsc{Functor} and \hsc{Monad} are already useful on their own: For example, we can use them to formalize generic Haskell monadic operations like \hsc{sequence} and \hsc{foldM}, and prove properties about them. Using the ordinary HOLCF Domain package with the right class constraints, we can also define higher-order type constructors:
\begin{equation*}
  \kwd{data}\:\hsc{Tree}(\tT,\tA) = \hsc{Tip} \mid \hsc{Node}\:\tA\:(\tT\cdot(\hsc{Tree}(\tT,\tA)))
\end{equation*}

This is good, but sooner or later, we will want to populate our constructor classes with some concrete instances. To show how a Tycon library user can define new functors and monads, we will now demonstrate the process with a recursive lazy list datatype.
\begin{equation*}
\kwd{data}\:\hsc{List}\:\tA = \hsc{Nil} \mid \hsc{Cons}\:\tA\:(\hsc{List}\:\tA)
\end{equation*}

The Domain package can handle this definition with no trouble. However, we do not want \hsc{List} to be an ordinary Isabelle type constructor, which can only appear in fully applied form. We want \hsc{List} as a first-class type constructor, i.e., an instance of class \hsc{Tycon}. We really want to write \emph{this} definition instead, which uses the type application operator:
\begin{equation*}
\kwd{data}\:\hsc{List}\cdot\tA = \hsc{Nil} \mid \hsc{Cons}\:\tA\:(\hsc{List}\cdot\tA)
\end{equation*}

The Tycon library now provides full automation for such type definitions, in the form of a new user-level type definition command. It works much like the HOLCF Domain package, and is implemented using much of the same code.

The process by which the Domain package defines new datatypes can be broken down roughly into four steps:
\begin{enumerate}
\item Define a deflation combinator, and use it to define a representable domain satisfying the domain equation.
\item Define constructors and related functions; generate theorems.
\item Define take function; derive induction rules.
\item Define map function; relate it to the deflation combinator.
\end{enumerate}
Defining a usable \hsc{Tycon} involves essentially the same four steps. However, some of the steps are adapted slightly to deal with the \hsc{Tycon} instance and the type application operator. We now describe how our new command completes each of the four steps to make \hsc{List} into a \hsc{Tycon} and \hsc{Functor}.

\begin{enumerate}

\item

Just like the Domain package, it constructs a deflation as a least fixed-point, based on the recursive domain equation. However, instead of defining a type $\hsc{List}\:\tA$ directly from this deflation, it defines \hsc{List} as a singleton type, and makes it an instance of class \hsc{Tycon}. The constructed deflation is used to define $\TC{\hsc{List}}$.
\begin{align}
\TC{\hsc{List}}(a) & = (\mu\hair{t}.\:\hsone_\D \oplus_\D (a_{\bot_\D} \otimes_\D t_{\bot_\D})) \\
\label{eq:repiso}
\REP{\hsc{List}\cdot\tA} & = \REP{\hsone \oplus (\tA_\bot \otimes (\hsc{List}\cdot\tA)_\bot)}
\end{align}
By unfolding the fixed-point, the desired domain equation \eqref{eq:repiso} is derived. It then follows that the coercions \hsc{absList} and \hsc{repList}, defined as shown here, form an isomorphism.
\begin{align*}
\hsc{absList}_\tA & = \hscoerce_{(\hsone \oplus (\tA_\bot \otimes (\hsc{List}\cdot\tA)_\bot), \hsc{List}\cdot\tA)} \\
\hsc{repList}_\tA & = \hscoerce_{(\hsc{List}\cdot\tA, \hsone \oplus (\tA_\bot \otimes (\hsc{List}\cdot\tA)_\bot))}
\end{align*}

\item

Using these isomorphism theorems, a component of the Domain package is called to generate the multitude of definitions and theorems related to the constructors \hsc{Nil} and \hsc{Cons}. This step works exactly the same as with ordinary domain definitions.

\item

A call to another Domain package component generates a chain of \hsc{listTake} functions:
\begin{align*}
  &
  \begin{alignedat}{3}
    & \hsc{listTake} :: \hsc{Nat} \to \hsc{List}\cdot\tA \to \hsc{List}\cdot\tA
  \end{alignedat}
  \\[-\jot] &
  \begin{alignedat}{3}
    & \hsc{listTake}\:0 && \hsc{xs} &&= \bot \\[-\jot]
    & \hsc{listTake}\,(n+1)\:&&\hsc{Nil} &&= \hsc{Nil}\\[-\jot]
    & \hsc{listTake}\,(n+1)\:&&(\hsc{Cons}\:x\:\hsc{xs})
    &&= \hsc{Cons}\:x\:(\hsc{listTake}\:n\:\hsc{xs})
  \end{alignedat}
\end{align*}
By reasoning about the deflation agreement relation $(\isodefl)$, we can show
$\bigsqcup_n \hsc{listTake}\:n = \hsid$ from the definitions of \hsc{listTake} and the deflation combinator. From this, the approximation lemma \cite{Hutton2001} and induction rules are then derived, just as they are in the Domain package.

\item

The final step is to instantiate the \hsc{Functor} class. The $\fmapU$ function is defined in a stylized way, which exactly matches the structure of the definition of $\TC{\hsc{List}}$.
\begin{multline}
\fmapU^\hsc{List}f = (\mu\hair{t}.\:\hsc{absList} \circ{} \\[-\jot]
\mapSum(\hsid_\hsone, \mapProd(\mapLift\:f, \mapLift\:t)) \circ \hsc{repList})
\end{multline}
The Domain package would normally generate the same definition, but would define it as a separate constant \hsc{mapList}.

\begin{figure}
\begin{mathpar}
\hsid_\tA \isodefl \REP{\tA}
\and
\inferrule
    {f \isodefl d \\ \REP{\tA} = \REP{\hair\tB}}
    {\hscoerce_{\tB,\tA} \circ f \circ \hscoerce_{\tA,\tB} \isodefl d}
\and
\inferrule{f \isodefl d}{\mapLift(f) \isodefl (d_{\bot_\D})}
\and
\inferrule{f_1 \isodefl d_1 \\ f_2 \isodefl d_2}{\mapSum(f_1, f_2) \isodefl (d_1 \oplus_\D d_2)}
\and
\inferrule{f_1 \isodefl d_1 \\ f_2 \isodefl d_2}{\mapProd(f_1, f_2) \isodefl (d_1 \otimes_\D d_2)}
\and
\inferrule{f \isodefl d}{\fmap^\tT\,f \isodefl \TC{\tT}\,d}
\end{mathpar}
\caption{Agreement rules between map functions and deflations}
\label{fig:agreement}
\end{figure}

The \hsc{Functor} class requires a proof of the agreement law $\fmapU^\hsc{List}d \isodefl \TC{\hsc{List}}\,d$. Because the definitions of $\fmapU^\hsc{List}$ and $\TC{\hsc{List}}$ have the same structure, the proof can be discharged using a collection of structural rules, some of which are listed in Fig.~\ref{fig:agreement}. The Domain package maintains this list of rules for use in its own internal proofs \cite[\S6.6]{holcf11}.

\end{enumerate}

It is not always possible to automatically prove the functor composition law: For some strict datatypes, the composition law can fail when used with non-strict functions. To avoid this difficulty, we split off a separate \hsc{Prefunctor} superclass that asserts only the identity law. Our new command can then always succeed in generating a \hsc{Prefunctor} instance for each new datatype; we leave it to the user instantiate the \hsc{Functor} class by proving the composition law separately.

For the \hsc{List} type constructor, composition can be proved using the ordinary HOLCF technique of induction over the datatype.

\paragraph{Further class instantiations.}

Compared to \hsc{Tycon} and \hsc{Functor}, instantiations of subclasses like \hsc{FunctorPlus} and \hsc{Monad} are relatively straightforward. We write definitions of $(\hsappU)$, $\univ{\hsc{return}}$, and $(\hsbindU)$ using ordinary user-level methods: the standard Isabelle definition command for non-recursive functions, and the HOLCF Fixrec package \cite{holcf11} for the recursive ones. The class axioms for these subclasses are all ordinary equations, so they can be proved using ordinary techniques like induction.

\paragraph{Transferring theorems.}

We now have a type constructor \hsc{List} with instances of the \hsc{Functor}, \hsc{FunctorPlus}, and \hsc{Monad} classes. This means that we can use the polymorphic functions $\fmap$, $(\hsapp)$, \hsc{return}, and $(\hsbind)$ at type $\hsc{List}\cdot\tA$. We can also apply any generic theorems from those classes to the \hsc{List} type.

However, we do not have any \hsc{List}-specific theorems about the polymorphic functions yet. For example, if $\hsc{Cons}\;x\;\hsc{xs} \hsappU \hsc{ys} = \hsc{Cons}\;x\;(\hsc{xs} \hsappU \hsc{ys})$ is one of the defining equations for $(\hsappU)$, we should like to have a version of this theorem for $(\hsapp)$ as well.

To obtain the polymorphic versions of such lemmas, we need to do a transfer process, much like we did with Theorem~\ref{thm:functor-composition} and for the class axioms in Sec.~\ref{sec:functorplus}. The proofs can generally be completed with one call to the simplifier, using a collection of simplification rules for coercions. To transfer theorems that mention \hsc{Nil} or \hsc{Cons}, we must first prove some additional simplification rules stating that \hsc{coerce} commutes with those data constructors. These proofs are also simple, and potentially could be generated automatically.

%%%%%%%%%%%%%%%%%%%%%%%%%%%%%%%%%%%%%%%%%%%%%%%%%%
\section{Verifying monad transformers}
\label{sec:monad-transformers}

In addition to simple type constructors like \hsc{List}, the Tycon library can also be used to define \hsc{Tycon} instances with additional type parameters, some of which may be type constructors themselves. In particular, this means that we can define a monad transformer---i.e., a monad that is parameterized by another inner monad.

The resumption monad transformer was covered in our previous work \cite{HMW2005}, but we have some improvements here. With the improved class definitions and better proof automation, we can now prove more with less effort: In addition to instantiating the monad class, we also proceed to define an interleaving operator and prove laws about it.

The new automation provided by the Tycon library has made it easier to test out definitions of new type constructors. Experimentation with the error and writer monad transformers has revealed that neither one truly preserves the monad laws. However, we have also found that the monad laws for both of those types actually \emph{are} preserved for values constructible from standard operations. That is, it is possible to view each as an abstract datatype whose operations maintain an invariant; in this abstract view, each one actually does form a lawful monad.

\subsection{Resumption monad transformer}
\label{sec:ResT}

The resumption monad transformer \cite{Papaspyrou2001} augments an inner monad with the ability to suspend, resume, and interleave threads of computations. The Haskell definitions for the resumption monad transformer are shown in Fig.~\ref{fig:resT}. (Note that although we call it a monad transformer, the \hsc{Monad} instance only requires $\tT$ to be a functor.)

\newcommand{\apRT}{\mathbin{\circledast}}

\begin{figure}
\begin{align*}
  & \kwd{data}\;\hsc{ResT}\;\tT\;\tA = \hsc{Done}\;\tA \mid \hsc{More}\;(\tT\;(\hsc{ResT}\;\tT\;\tA))
  \\[\jot]
  & \kwd{instance}\;(\hsc{Functor}\;\tT) \To \hsc{Monad}\;(\hsc{ResT}\;\tT)\;\kwd{where}
  \\[-\jot]
  & \hspace{8pt}
  \begin{alignedat}{2}
    & \hsc{return}\;x &&= \hsc{Done}\;x \\[-\jot]
    & \hsc{Done}\;x \hsbind k &&= k\;x \\[-\jot]
    & \hsc{More}\;m \hsbind k &&= \hsc{More}\;(\fmap\;(\lambda{r}.\;r \hsbind k)\;m)
  \end{alignedat}
\end{align*}
\caption{Haskell definition of \hsc{ResT} monad transformer}
\label{fig:resT}
\end{figure}

The constructor $\hsc{Done}\:x$ represents a computation that has run to completion, yielding the result $x$. $\hsc{More}\:c$ represents a suspended computation that still has more work to do: When $c$ is evaluated, it may produce some side effects (according to the monad $\tT$) and eventually returns a new resumption of type $\hsc{ResT}\:\tT\cdot\tA$. Resumptions are a bit like threads in a cooperative multitasking system: A running thread may either terminate ($\hsc{Done}\;x$) or voluntarily yield to the operating system, waiting to be resumed later ($\hsc{More}\;c$).

We formalize the Haskell type $\hsc{ResT}\;\tT\;\tA$ as $\hsc{ResT}\;\tT\cdot\tA$ in our library. The type constructor definition generates an $\fmap$ function satisfying these rules:
\begin{align*}
  \fmap\:f\:(\hsc{Done}\;x) &= \hsc{Done}\:(f\:x) \\[-\jot]
  \fmap\:f\:(\hsc{More}\;m) &= \hsc{More}\:(\fmap\:(\fmap\:f)\:m)
\end{align*}

From the low-level principle of take induction, we derive a high-level induction rule for type $\hsc{ResT}\:\tT\cdot\tA$: % \cite[Ch.~7]{holcf11}:
\begin{equation}
\inferrule
{\mathrm{admissible}(P)
\\ \forall{x}.\,P(\hsc{Done}\:x)
\\ P(\bot)
\\ \forall{m}\,{f}.\,(\forall{r}.\,P(f\:r)) \Longrightarrow P(\hsc{More}\:(\fmap\:f\:m))}
{\forall{r}.\,P(r)}
\end{equation}
We then proceed to instantiate the \hsc{Monad} class for $\hsc{ResT}\;\tT$; the proofs of the monad laws are all proved using the high-level induction rule. With this class instance, we have shown that \hsc{ResT} is a valid monad transformer.

\begin{figure}
\begin{align*}
  & \hair(\apRT) :: (FunctorPlus\;\tT) \To {} \\[-\jot]
  & \hspace{16pt}
  \hsc{ResT}\;\tT\;(\tA\to\tB) \to \hsc{ResT}\;\tT\;\tA \to \hsc{ResT}\;\tT\;\tB
  \\[-\jot]
  &
  \begin{alignedat}{3}
    & \hsc{Done}\;f &&\apRT \hsc{Done}\;x &&= \hsc{Done}\;(f\;x)
    \\[-\jot]
    & \hsc{Done}\;f &&\apRT \hsc{More}\;v &&=
    More\;(\fmap\;(\lambda{r}.\;\hsc{Done}\;f \apRT r)\;v)
    \\[-\jot]
    & \hsc{More}\;u &&\apRT \hsc{Done}\;x &&=
    More\;(\fmap\;(\lambda{r}.\;r \apRT \hsc{Done}\;x)\;u)
    \\[-\jot]
    & \hsc{More}\;u &&\apRT \hsc{More}\;v &&=
    \hsc{More}\;(\fmap\;(\lambda{r}.\; \hsc{More}\;u \apRT r)\;v
    \\[-\jot]
    &&&&& \hspace{22pt} {} \hsapp \fmap\;(\lambda{r}.\;r \apRT More\;v)\;u)
  \end{alignedat}
\end{align*}
\caption{Haskell definition of interleaving operator for \hsc{ResT}}
\label{fig:interleave}
\end{figure}

Some new features of our library are nicely demonstrated by the definition and verification of an interleaving operator for resumptions \cite{Papaspyrou2001}. The Haskell definition can be seen in Fig.~\ref{fig:interleave}. If both arguments are \hsc{Done}, then we combine the results and terminate.\footnote{We combine the results with function application so that we get an applicative functor; in other contexts a pair constructor might make more sense.} While either argument is \hsc{More}, we nondeterministically choose one such argument, run it for one step, and then recurse. Note that the definition uses a \hsc{FunctorPlus} class constraint---a type class whose formalization was made possible by the new Tycon library.

It turns out that $(\apRT)$ satisfies all the laws of an applicative functor \cite{McBride2008}. The trickiest to prove is the associativity law:
\begin{equation}
  \hsc{Done}\;(\circ) \apRT u \apRT v \apRT w = u \apRT (v \apRT w)
\end{equation}

The proof proceeds by nested inductions on $u$, $v$, and $w$; subproofs for the non-trivial cases rely on the naturality and associativity laws from the \hsc{FunctorPlus} class. A formalization of the same theorem was presented in the author's Ph.D thesis \cite{holcf11}, although there it was defined with a fixed inner monad. This version is more general and more abstract. We assume exactly what we need to about the type constructor $\tT$, nothing more.

\subsection{Error monad transformer}

The error monad transformer appears in Andy Gill's mtl library, inspired by Jones \cite{Jones1995}. It is simply a composition of the inner monad with an ordinary error monad. The Haskell definition of the \hsc{Error} monad that we use is shown in Fig.~\ref{fig:error}. It is parameterized by an extra type $\tE$, the type of error values.

\begin{figure}
\begin{align*}
  & \kwd{data}\;\hsc{Error}\;\tE\;\tA =
  \hsc{Err}\;\tE \mid \hsc{Ok}\;\tA
  \\[\jot]
  & \kwd{instance}\:\hsc{Functor}\:(\hsc{Error}\:\tE)\:\kwd{where}
  \\[-\jot]
  & \hspace{8pt}
  \begin{alignedat}{2}
    & \fmap\:f\:(\hsc{Err}\;e) &&= \hsc{Err}\;e \\[-\jot]
    & \fmap\:f\:(\hsc{Ok}\;a) &&= \hsc{Ok}\;(f\:a)
  \end{alignedat}
  \\[\jot]
  & \kwd{instance}\:\hsc{Monad}\:(\hsc{Error}\:\tE)\:\kwd{where}
  \\[-\jot]
  & \hspace{8pt}
  \begin{alignedat}{2}
    & \hsc{return}\;a &&= \hsc{Ok}\;a \\[-\jot]
    & \hsc{Err}\;e \hsbind k &&= \hsc{Err}\;e \\[-\jot]
    & \hsc{Ok}\;a \hsbind k &&= k\;a
  \end{alignedat}
\end{align*}
\caption{Haskell definition of \hsc{Error} monad}
\label{fig:error}
\end{figure}

\begin{figure}
\begin{align*}
  & \kwd{newtype}\:\hsc{ErrorT}\;\tE\;\tT\;\tA = \hsc{ErrorT}\:\{ \runET::\tT\:(\hsc{Error}\;\tE\;\tA) \}
  \\[\jot]
  & \kwd{instance}\:(\hsc{Monad}\:\tT) \To \hsc{Functor}\:(\hsc{ErrorT}\:\tE\:\tT)\:\kwd{where}
  \\[-\jot]
  & \hspace{8pt} \fmap\:f\:(\hsc{ErrorT}\:t) = \hsc{ErrorT}\,(\fmap\:(\fmap\:f)\:t)
  \\[\jot]
  & \kwd{instance}\:(\hsc{Monad}\:\tT) \To \hsc{Monad}\:(\hsc{ErrorT}\:\tE\:\tT)\:\kwd{where}
  \\[-\jot]
  & \hspace{8pt}
  \begin{alignedat}{2}
    & \hsc{return}\:a &&= \hsc{ErrorT}\,(\hsc{return}\,(\hsc{Ok}\;a)) \\[-\jot]
    & m \hsbind k &&= \hsc{ErrorT}\,(\runET\:m \hsbind \lambda\,n. \\[-\jot]
    & && \hspace{24pt} \mathrm{case}\;n\;\mathrm{of}\;
    \begin{aligned}[t]
      \hsc{Err}\;e &\to \hsc{return}\,(\hsc{Err}\;e) \\[-\jot]
      \hsc{Ok}\;a &\to \runET\,(k\;a))
    \end{aligned}
  \end{alignedat}
\end{align*}
\caption{Haskell definition of \hsc{ErrorT} monad transformer}
\label{fig:errorT}
\end{figure}

We define an instance $\hsc{Monad}\:(\hsc{Error}\:\tE)$ using the standard procedure outlined in Sec.~\ref{sec:instantiation}. The formal proofs of the monad laws proceed as expected. The resulting error monad type satisfies the following domain equation:
\begin{equation}
\REP{\hsc{Error}\:\tE\cdot\tA} = \REP{\tE_\bot\oplus\hair\tA_\bot}
\end{equation}

Using the error monad type, we can now proceed to define the error monad transformer.
We follow the Haskell definitions from Fig.~\ref{fig:errorT}, defining \hsc{ErrorT} as a newtype (i.e., a datatype with a single strict constructor).
\begin{align*}
& \kwd{newtype}\:(\hsc{Monad}\:\tT) \To \hsc{ErrorT}(\tE,\tT)\cdot\tA \\[-\jot]
& \hspace{8pt} = \hsc{ErrorT}\:\{\,\runET :: \tT\cdot(\hsc{Error}\:\tE\cdot\tA)\,\}
\end{align*}

The HOLCF error transformer type satisfies the following domain equation. Note that as a newtype, the right-hand side of its domain equation is not lifted.
\begin{equation}
\REP{\hsc{ErrorT}(\tE,\tT)\cdot\tA} = \REP{\tT\cdot(\hsc{Error}\:\tE\cdot\tA)}
\end{equation}

Building an instance of $\hsc{Functor}\,(\hsc{ErrorT}(\tE,\tT))$ in the standard way, we get a definition of $\fmap$ that satisfies the following rule, as we would expect:
\begin{multline}
\fmap^{(\hsc{ErrorT}(\tE,\tT))} f\,(\hsc{ErrorT}\:t) =\\[-\jot]
%\fmap\:f\:(\hsc{ErrorT}\:t) =\\
\hsc{ErrorT}\,(\fmap^\tT (\fmap^{\hsc{Error}(\tE)} f)\:t)
\end{multline}

\paragraph{Problems with monad instance.}

Unfortunately, we run into difficulty when trying to prove an instance of $\hsc{Monad}\,(\hsc{ErrorT}(\tE,\tT))$. Not all of the class axioms are provable. The \hsc{Monad} class will not let us define constants \hsc{return} and $(\hsbind)$ that do not satisfy the laws, so instead we define the \hsc{return} and $(\hsbind)$ from Fig.~\ref{fig:errorT} as separate constants \hsc{unitET} and \hsc{bindET}. These and other HOLCF definitions for the error monad transformer type are shown in Fig.~\ref{fig:ET}.

\begin{figure*}
\begin{align*}
& \hsc{unitET}::(\hsc{Monad}\:\tT)\To\tA\to\hsc{ErrorT}(\tE,\tT)\cdot\tA
\\[-\jot]
& \hsc{unitET}\:a = \hsc{ErrorT}\,(\hsc{return}^\tT\,(\hsc{Ok}\;a))
\\[4pt]
& \hsc{bindET}::(\hsc{Monad}\:\tT)\To\hsc{ErrorT}(\tE,\tT)\cdot\tA\to(\tA\to\hsc{ErrorT}(\tE,\tT)\cdot\tB)\to\hsc{ErrorT}(\tE,\tT)\cdot\tB
\\[-\jot]
& \hsc{bindET}\:m\:k = \hsc{ErrorT}\,(\runET\:m \hsbind^\tT \lambda\,{x}.\:
\mathrm{case}\;x\;\mathrm{of}\;\hsc{Err}\;e \to \hsc{return}^\tT\,(\hsc{Err}\;e); \hsc{Ok}\;a \to \runET\:(k\;a))
\\[4pt]
& \hsc{liftET}::(\hsc{Monad}\:\tT) \To \tT\cdot\tA\to\hsc{ErrorT}(\tE,\tT)\cdot\tA
\\[-\jot]
& \hsc{liftET}\:t = \hsc{ErrorT}\,(\fmap^\tT\,\hsc{Ok}\:t)
\\[4pt]
& \hsc{throwET}::(\hsc{Monad}\:\tT)\To\tE\to\hsc{ErrorT}(\tE,\tT)\cdot\tA
\\[-\jot]
& \hsc{throwET}\:e = \hsc{ErrorT}\,(\hsc{return}^\tT\,(\hsc{Err}\;e))
\\[4pt]
& \hsc{catchET}::(\hsc{Monad}\:\tT)\To\hsc{ErrorT}(\tE,\tT)\cdot\tA\to(\tE\to\hsc{ErrorT}(\tE,\tT)\cdot\tA)\to\hsc{ErrorT}(\tE,\tT)\cdot\tA
\\[-\jot]
& \hsc{catchET}\:m\:h = \hsc{ErrorT}\:(\runET\:m \hsbind^\tT \lambda\,{x}.\:
\mathrm{case}\;x\;\mathrm{of}\;\hsc{Err}\;e \to \runET\,(h\;e); \hsc{Ok}\;a \to \hsc{return}^\tT\,(\hsc{Ok}\;a))
\end{align*}
\caption{Isabelle definitions of error monad transformer operations}
\label{fig:ET}
\end{figure*}

Using this collection of non-overloaded constants, we can examine in detail the situations where the laws fail. In fact, most of the expected laws, e.g. the left unit law, do hold in general. All of the lemmas shown below can be proven by showing that $\runET$ applied to each side yields the same value.
\begin{gather}
\hsc{bindET}\:(\hsc{unitET}\:a)\:k = k\;a \\
\hsc{catchET}\:(\hsc{throwET}\:e)\:h = h\;e \\
\hsc{bindET}\:(\hsc{throwET}\:e)\:k = \hsc{throwET}\:e \\
\hsc{catchET}\:(\hsc{unitET}\:a)\:h = \hsc{unitET}\:a \\
\hsc{liftET}\:(\hsc{return}^\tT\:a) = \hsc{unitET}\:a \\
\hsc{liftET}\:(t \hsbind^\tT k) = \hsc{bindET}\:(\hsc{liftET}\:t)\:(\hsc{liftET} \circ k)
\end{gather}
A more involved proof shows that associativity also holds for \hsc{bindET}.

\begin{theorem}
\label{thm:errorT-assoc}
The error monad transformer satisfies the monad associativity law.
\begin{equation*}
\hsc{bindET}\:(\hsc{bindET}\:m\;h)\:k
= \hsc{bindET}\:m\:(\lambda{a}.\:\hsc{bindET}\:(h\;a)\:k)
\end{equation*}
\end{theorem}
\begin{proof}
Let $R(k)$ abbreviate the lambda expression in the definition of \hsc{bindET}, so that $\hsc{runET}\:(\hsc{bindET}\:m\:k) = \hsc{runET}\:m \hsbind R(k)$. Also note that $R(k)$ is strict. The proof then proceeds by applying \hsc{runET} to both sides of the equation. After simplification, we have:
\begin{multline*}
(\hsc{runET}\:m \hsbind R(h)) \hsbind R(k) \\[-\jot]
= \hsc{runET}\:m \hsbind R(\lambda{a}.\:\hsc{bindET}\:(h\:a)\:k)
\end{multline*}
After rewriting the left-hand side with the associativity law, both sides have the form $\hsc{runET}\:m \hsbind f$. It then suffices to show that the functions on both sides are equal for all arguments:
\begin{equation*}
\forall{x}.\: R(h)\:x \hsbind R(k) = R(\lambda{a}.\:\hsc{bindET}\:(h\:a)\:k)\:x
\end{equation*}
We proceed by cases on $x$. If $x = \bot$, then using Theorem~\ref{thm:strict-bind} we see that both sides reduce to $\bot$. If $x = \hsc{Err}\;e$, then both sides reduce to $\hsc{return}^\tau\,(\hsc{Err}\;e)$. Finally, if $x = \hsc{Ok}\;a$, then both sides evaluate to $\hsc{runET}\,(h\:a) \hsbind R(k)$.
\end{proof}

On the other hand, the right unit monad law is not satisfied in general. Unless the inner monad $\tT$ has a strict \hsc{return} function, $m = \hsc{ErrorT}\,(\hsc{return}\:\bot)$ is a counterexample to the right unit law.

\begin{theorem}
\label{thm:errorT-right-unit}
The error monad transformer satisfies the right unit law if and only if the inner monad has a strict \hsc{return}.
\begin{equation*}
(\forall{m}.\:\hsc{bindET}^\tT\,m\;\hsc{unitET}^\tT = m) \Longleftrightarrow
(\hsc{return}^\tT\,\bot = \bot)
\end{equation*}
\end{theorem}

\begin{proof}
Case $(\Longrightarrow)$: If we instantiate $m = \hsc{ErrorT}\:(\hsc{return}\:\bot)$, then the equation reduces to $\bot = \hsc{return}\:\bot$. Case $(\Longleftarrow)$: As above, let $R(k)$ abbreviate the lambda expression in the definition of \hsc{bindET}. We proceed to show $\hsc{bindET}\:m\:\hsc{unitET} = m$ by applying \hsc{runET} to both sides. After simplification, we get:
\begin{equation*}
\hsc{runET}\:m \hsbind R(\hsc{unitET}) = \hsc{runET}\:m
\end{equation*}

After expanding the right-hand side with the right unit law, both sides have the form $\hsc{runET}\:m \hsbind f$. It then suffices to show that the functions on both sides are equal for all arguments:
\begin{equation*}
\forall{x}.\:R(\hsc{unitET})\:x = \hsc{return}\;x
\end{equation*}

If $x = \bot$, then the equation reduces to $\bot = \hsc{return}\:\bot$, which we solve by assumption. In case $x = \hsc{Err}\;e$ or $x = \hsc{Ok}\;a$, then the equation reduces to a trivial equality.
\end{proof}

We could prove a monad class instance for the error transformer by creating a subclass for monads-with-strict-return, and putting a stronger constraint on type $\tT$:
\begin{equation*}
\kwd{instance}\:(\hsc{StrictMonad}\:\tT) \To \hsc{Monad}\:(\hsc{ErrorT}(\tE,\tT))
\end{equation*}
However, this is not very useful in practice, because most monads do not have a strict \hsc{return} function (although there are a few that do, e.g. the Identity monad and some varieties of powerdomains).

\paragraph{Data abstraction to the rescue.}

It turns out that it is impossible to construct the offending value $\hsc{ErrorT}\:(\hsc{return}\:\bot)$ using only the standard operations listed in Fig.~\ref{fig:ET}. Furthermore, we can show that for all values constructible using those operations, the monad laws do always hold. This means that when viewed as an abstract datatype, we could still consider \hsc{ErrorT} to be a valid monad.

We define an inductive set $\INV$ that includes all values that can be constructed with functions in the abstract interface (see Fig.~\ref{fig:invET}). We must also include rules for $\bot$ and lubs, to ensure that the set $\INV$ is a pcpo: In Haskell it is possible to define recursive values (i.e., least fixed-points) at any type, abstract or not.

\begin{figure}
\begin{mathpar}
\inferrule{}{\hsc{unitET}\:a \in \INV}
\and
\inferrule{m \in \INV \\ \forall{a}.\:k\;a \in \INV}{\hsc{bindET}\:m\:k \in \INV}
\and
\inferrule{}{\hsc{throwET}\:e \in \INV}
\and
\inferrule{m \in \INV \\ \forall{e}.\:h\;e \in \INV}{\hsc{catchET}\:m\:h \in \INV}
\and
\inferrule{}{\hsc{liftET}\:t \in \INV}
\and
\inferrule{}{\bot \in \INV}
\and
\inferrule{\forall{i}.\: m_i \in \INV}{\textstyle{\bigsqcup_i m_i} \in \INV}
\end{mathpar}
\caption{Inductive invariant based on \hsc{ErrorT} abstract interface}
\label{fig:invET}
\end{figure}

Finally, we can prove a restricted form of the right unit law by induction on $\INV$. The proof is straightforward, and uses techniques similar to those used for Theorems \ref{thm:errorT-assoc} and \ref{thm:errorT-right-unit}.
\begin{equation}
m \in \INV \implies \hsc{bindET}\;m\;\hsc{unitET} = m
\end{equation}

Besides using an inductive set, there is another, more direct way of defining the invariant. We can define $\INV$ simply as the set of all values satisfying the right unit law:
\begin{equation}
  \INV = \{\,m \mid \hsc{bindET}\;m\;\hsc{unitET} = m\,\}
\end{equation}
It turns out that $(\lambda{m}.\:\hsc{bindET}\;m\;\hsc{unitET})$ is actually a deflation, of which this version of $\INV$ is the corresponding set. (The reader may wish to verify that it is idempotent and below $\hsid$.)

Conveniently, we are already using deflations as our model of types. Therefore, we can use this deflation to define a new representable subtype of $\hsc{ErrorT}(\tE,\tT)\cdot\tA$ that is isomorphic to the set $\INV$. The representation of the new type $\hsc{ErrorT'}(\tE,\tT)\cdot\alpha$ as a deflation is therefore as follows:
\begin{equation}
\REP{\hsc{ErrorT'}(\tE,\tT)\cdot\tA}
= \hsemb \circ (\lambda\,m.\:\hsc{bindET}\:m\;\hsc{unitET}) \circ \hsprj
\end{equation}

We have implemented such a type definition using the Tycon library, and proven a \hsc{Monad} class instance for it. However, we do not yet have a principled technique for transferring definitions or theorems between the \hsc{ErrorT} and \hsc{ErrorT'} types, so working with such subtypes is impractical for casual users. Exploring ways to automate this process will be an area for future research.

%%%%%%%%%%%%%%%%%%%%%%%%%%%%%%%%%%%%%%%%%%%%%%%%%%
\subsection{Writer monad transformer}

The writer monad allows a program to output a string (or more generally, any \hsc{Monoid} type) along with its ordinary result \cite{Jones1995}. The bind operation of the monad concatenates the strings output by each sub-computation. The writer monad transformer composes the writer monad with an inner monad, extending the inner monad with a string output capability. The Haskell definitions are shown in Fig.~\ref{fig:writerT}.

\begin{figure}
\begin{align*}
  & \kwd{class}\:\hsc{Monoid}\:\tW\:\kwd{where} \\[-\jot]
  & \hspace{8pt}
  \begin{alignedat}{2}
    & \mempty &&:: \tW \\[-\jot]
    & (\mappend) &&:: \tW \to \tW \to \tW
  \end{alignedat}
  \\[\jot]
  & \kwd{data}\;\hsc{Writer}\;\tW\;\tA = \hsc{Result}\;\tW\;\tA
  \\[\jot]
  & \kwd{newtype}\:\hsc{WriterT}\:\tW\:\tT\:\tA
  = \hsc{WriterT}\:\{ \hsc{runWT} :: \tT\:(\hsc{Writer}\:\tW\:\tA) \}
  \\[\jot]
  & \kwd{instance}\:(\hsc{Monoid}\:\tW, \hsc{Monad}\:\tT) \To \hsc{Monad}\:(\hsc{WriterT}\:\tW\:\tT)\:\kwd{where}
  \\[-\jot]
  & \hspace{8pt}
  \begin{alignedat}{2}
    & \hsc{return}\;a &&= \hsc{WriterT}\,(\hsc{return}\,(\hsc{Result}\;\mempty\;a)) \\[-\jot]
    & m \hsbind k &&= \hsc{WriterT}\,(
    \begin{aligned}[t]
      & \hsc{runWT}\:m \hsbind \lambda(\hsc{Result}\;w_1\;a). \\[-\jot]
      & \hsc{runWT}\:(k\;a) \hsbind \lambda(\hsc{Result}\;w_2\;b). \\[-\jot]
      & \hsc{return}\:(\hsc{Result}\;(w_1 \mappend w_2)\;b))
    \end{aligned}
  \end{alignedat}
  \\[\jot]
  & \hsc{tell} :: \tW \to \hsc{WriterT}\;\tW\;\tT\;() \\[-\jot]
  & \hsc{tell}\;w = \hsc{WriterT}\;(\hsc{return}\;(\hsc{Result}\;w\;()))
  \\[\jot]
  & \hsc{listen} :: \hsc{WriterT}\;\tW\;\tT\;\tA \to \hsc{WriterT}\;\tW\;\tT\;(\hsc{Writer}\;\tW\;\tA) \\[-\jot]
  & \hsc{listen}\;m = \hsc{WriterT}\:(
  \begin{aligned}[t]
    & runWT\;m \hsbind \lambda(\hsc{Result}\;w\;a). \\[-\jot]
    & \hsc{Result}\;w\;(\hsc{Result}\;w\;a))
  \end{aligned}
\end{align*}
\caption{Haskell definition of \hsc{WriterT} monad transformer}
\label{fig:writerT}
\end{figure}

The Haskell \hsc{Monoid} class has a set of customary axioms: Instances should ensure that $(\mappend)$ is associative, with $\mempty$ as the identity element, so $\mempty \mappend x = x \mappend \mempty = x$. Note that \hsc{Monoid} is not a constructor class, so we can formalize it as an ordinary Isabelle type class.

The formalization of the writer monad transformer works out in almost exactly the same way as the error monad transformer: The type definitions and \hsc{Functor} instances work fine, but the monad instance fails because neither the left nor the right unit law holds in general. To reason about \hsc{return} and \hsc{bind} without a \hsc{Monad} class instance, we define functions \hsc{unitWT} and \hsc{bindWT} according to the definitions in Fig.~\ref{fig:writerT}.

\begin{theorem}
\label{thm:writerT-right-unit}
The writer monad transformer satisfies the right unit law if and only if the inner monad has a strict \hsc{return}.
\begin{equation*}
(\forall{m}.\;\hsc{bindWT}^\tT\:m\;\hsc{unitWT}^\tT = m) \Longleftrightarrow
(\hsc{return}^\tT\,\bot = \bot)
\end{equation*}
\end{theorem}
\begin{proof}
Similar to Theorem~\ref{thm:errorT-right-unit}. In the case that \hsc{return} is not strict, instantiating $m = \hsc{WriterT}\:(\hsc{return}\:\bot)$ gives the counterexample.
\end{proof}

\begin{theorem}
\label{thm:writerT-left-unit}
The writer monad transformer satisfies the left unit law if and only if the inner monad has a strict \hsc{return}.
\begin{equation*}
(\forall{x}\:{k}.\;\hsc{bindWT}^\tT\:(\hsc{unitWT}^\tT\:x)\;k = k\;x) \Longleftrightarrow
(\hsc{return}^\tT\,\bot = \bot)
\end{equation*}
\end{theorem}
\begin{proof}
Similar to Theorem~\ref{thm:writerT-right-unit}. In case \hsc{return} is not strict, instantiating $k = \lambda{x}.\:\hsc{WriterT}\:(\hsc{return}\:\bot)$ gives the counterexample.
\end{proof}

As with the error monad transformer, we can define a subset of type consisting of those values that satisfy the right unit law:
\begin{equation}
  \INV = \{\,m \mid \hsc{bindWT}\;m\;\hsc{unitWT} = m\,\}
\end{equation}
It is straightforward to check that all writer transformer operations preserve this invariant, including \hsc{unitWT}, \hsc{bindWT}, and the formalized versions of \hsc{tell} and \hsc{listen}.

The reader may verify that the function $\lambda{m}.\;\hsc{bindWT}\:m\:\hsc{unitWT}$ is indeed a deflation. But we are not quite done showing that the subtype defined by $\INV$ is a monad: Because the left unit law does not hold universally for the writer transformer, we must also verify that all values in $\INV$ satisfy the left unit law as well.
\begin{equation}
k\;x \in \INV \implies \hsc{bindWT}\;(\hsc{unitWT}\;x)\;k = k\;x
\end{equation}
Unfolding the definition of $\INV$, we see that it is sufficient to show $\hsc{bindWT}\;(\hsc{unitWT}\;x)\;k = \hsc{bindWT}\;(k\;x)\;\hsc{unitWT}$. This can easily be proven by applying \hsc{runWT} to both sides and simplifying.

In summary, we have seen that the writer monad transformer is not quite a true monad, because the type contains values that do not respect the monad laws. But when we view it as an abstract datatype, with an interface that exports only operations that preserve the datatype invariant, it is valid to treat it as a real monad.

%%%%%%%%%%%%%%%%%%%%%%%%%%%%%%%%%%%%%%%%%%%%%%%%%%
\section{Conclusions and related work}
\label{sec:conclusion}

The Tycon library for Isabelle/HOLCF is now available at the Archive of Formal Proofs \cite{AFP2012}. It allows users to define, reason about, and instantiate constructor classes with little effort. It models polymorphism using coercion from a universal domain, which allows it to work in ordinary higher-order logic.

A different domain-theoretic model of polymorphism is presented by Amadio and Curien \cite{amadio+curien}. Here, polymorphic functions are modeled as functions from types (i.e. deflations) to values. However, this model allows non-parametric polymorphic functions that depend non-trivially on the type argument. Also, building a Tycon library around this model would also require users to write explicit type abstractions and applications when instantiating constructor classes; it is not clear whether this would be practical for users.

Sozeau and Oury~\cite{Sozeau2008} have recently developed a type class mechanism for the Coq theorem prover. Coq has a powerful dependent type system that allows reasoning about type constructors, first-class polymorphic values and type quantification. They define a monad class, including monad laws. Their system has the capability to formalize the whole monad class hierarchy, and it appears that it could be used to verify monad transformers; however, we are unaware of any published work in that direction.

Formalizing monad transformers in Coq does have some limitations compared to the Tycon library. For example, Coq does not accept the type definition of the resumption monad transformer: To ensure the strict positivity requirement, indirect recursion can only be used with known type constructors, not a monad parameter. Another difference (not necessarily a limitation) is that Coq is a logic of terminating functions, and does not include notions of bottoms, strictness, or partial values. Results proved in such a logic must be interpreted differently.

Our earlier formalization of axiomatic constructor classes \cite{HMW2005} could express many of the same definitions as the current work. However, it did not provide as many features or as much automation for users. Instead of naturality laws, it used a deflation membership relation, written $x ::: d$, to express the fact that polymorphic functions had the right type. For example, the \hsc{Monad} class there had a rule for $\returnU$ that stated $\forall{x}:::d.\; \returnU\;x ::: \TC{\tau}(d)$. Transfer proofs to establish polymorphic laws were lengthy and unprincipled, making subclass definitions impractical for users. Automation for instantiating the \hsc{Functor} and \hsc{Monad} classes was present, but it required users to define $\fmapU$ on a separate copy of the datatype first. Users were then left without a good way to transfer properties to the new \hsc{Tycon} version of the type.

Automation for theorem transfer in the Tycon library is much smoother than it was in our earlier work, but there is still room for further improvement. Currently we rely on a set of rewrite rules, which works well in practice so far. However, the behavior of such rewriting strategies is often hard to predict, the rules were assembled in an ad hoc fashion, and we have no convincing reason to trust that the method will work in all situations.

A better approach would be to use a more principled theorem transfer method, like the quotient packages developed recently by Homeier \cite{Homeier2005} and Kaliszyk \& Urban \cite{Kaliszyk2011}. For functions respecting an equivalence relation, theorems can be transferred from the underlying ``raw'' type to a quotient type. In HOLCF, type $\U$ could be considered as a raw type, with a representable type $\tA$ as a quotient type; the $\hsprj$ function induces an equivalence relation on $\U$. The naturality laws for operations on $\U$ could then serve as the respectfulness theorems required by the quotient package.

\acks

Thanks to John Matthews for many discussions about HOLCF which helped to develop the ideas in this paper. Thanks also to Jasmin Blanchette for reading an early draft and providing helpful comments.

\bibliographystyle{plain}
\bibliography{root}

\begin{thebibliography}{10}

\bibitem{amadio+curien}
Roberto~M. Amadio and Pierre-Louis Curien.
\newblock {\em Domains and Lambda-Calculi}.
\newblock Cambridge University Press, New York, NY, USA, 1998.

\bibitem{Gunter1990}
Carl~A. Gunter and Dana~S. Scott.
\newblock Semantic domains.
\newblock In Jan van Leeuwen, editor, {\em Handbook of Theoretical Computer
  Science, Volume B: Formal Models and Semantics (B)}, pages 633--674. MIT
  Press, 1990.

\bibitem{Homeier2005}
Peter~V. Homeier.
\newblock A design structure for higher order quotients.
\newblock In {\em Proceedings of the 18th International Conference on Theorem
  Proving in Higher Order Logics (TPHOLs '05)}, volume 3603 of {\em LNCS},
  pages 130--146. Springer-Verlag, 2005.

\bibitem{Huffman2009}
Brian Huffman.
\newblock A purely definitional universal domain.
\newblock In Stefan Berghofer, Tobias Nipkow, Christian Urban, and Makarius
  Wenzel, editors, {\em Proceedings of the 22nd International Conference on
  Theorem Proving in Higher Order Logics (TPHOLs '09)}, volume 5674 of {\em
  LNCS}, pages 260--275. Springer, 2009.

\bibitem{holcf11}
Brian Huffman.
\newblock {\em HOLCF '11: A Definitional Domain Theory for Verifying Functional
  Programs}.
\newblock Ph.{D}. thesis, Portland State University, 2012.

\bibitem{AFP2012}
Brian Huffman.
\newblock Type constructor classes and monad transformers.
\newblock {\em Archive of Formal Proofs}, June 2012.
\newblock \url{http://afp.sf.net/entries/Tycon.shtml}, Formal proof
  development.

\bibitem{HMW2005}
Brian Huffman, John Matthews, and Peter White.
\newblock Axiomatic constructor classes in {Isabelle/HOLCF}.
\newblock In Joe Hurd and Tom Melham, editors, {\em Proceedings of the 18th
  International Conference on Theorem Proving in Higher Order Logics (TPHOLs
  '05)}, volume 3603 of {\em LNCS}, pages 147--162. Springer, 2005.

\bibitem{Hutton2001}
Graham Hutton and Jeremy Gibbons.
\newblock The generic approximation lemma.
\newblock {\em Information Processing Letters}, 79:2001, 2001.

\bibitem{Jones1995}
Mark~P. Jones.
\newblock Functional programming with overloading and higher-order
  polymorphism.
\newblock In {\em First International Spring School on Advanced Functional
  Programming Techniques}, volume 925 of {\em LNCS}, B{\aa}stad, Sweden, May
  1995. Springer-Verlag.

\bibitem{Kaliszyk2011}
Cezary Kaliszyk and Christian Urban.
\newblock Quotients revisited for {I}sabelle/{HOL}.
\newblock In William~C. Chu, W.~Eric Wong, Mathew~J. Palakal, and Chih-Cheng
  Hung, editors, {\em Proc. of the 26th ACM Symposium on Applied Computing
  (SAC'11)}, pages 1639--1644. ACM, 2011.

\bibitem{McBride2008}
Conor McBride and Ross Paterson.
\newblock Applicative programming with effects.
\newblock {\em Journal of Functional Programming}, 18(1):1--13, 2008.

\bibitem{holcf99}
Olaf M\"uller, Tobias Nipkow, David~von Oheimb, and Oskar Slotosch.
\newblock {HOLCF = HOL + LCF}.
\newblock {\em Journal of Functional Programming}, 9:191--223, 1999.

\bibitem{isabelle-tutorial}
Tobias Nipkow, Lawrence~C. Paulson, and Markus Wenzel.
\newblock {\em Isabelle/HOL --- A Proof Assistant for Higher-Order Logic},
  volume 2283 of {\em LNCS}.
\newblock Springer, 2002.

\bibitem{Papaspyrou2001}
Nikolaos~S. Papaspyrou.
\newblock A resumption monad transformer and its applications in the semantics
  of concurrency.
\newblock In {\em Proceedings of the 3rd Panhellenic Logic Symposium}, Anogia,
  Greece, July 2001.

\bibitem{Sozeau2008}
Matthieu Sozeau and Nicolas Oury.
\newblock First-class type classes.
\newblock In Otmane~Ait Mohamed, C\'{e}sar Mu{\~{n}}oz, and Sofi\`{e}ne Tahar,
  editors, {\em Theorem Proving in Higher Order Logics, 21st International
  Conference (TPHOLs '08)}, volume 5170 of {\em LNCS}, pages 278--293.
  Springer, August 2008.

\bibitem{Wadler1989}
Philip Wadler.
\newblock Theorems for free!
\newblock In {\em Functional Programming Languages and Computer Architecture},
  pages 347--359. ACM Press, 1989.

\bibitem{Wenzel1997}
Markus Wenzel.
\newblock Type classes and overloading in higher-order logic.
\newblock In E.~Gunter and A.~Felty, editors, {\em Proceedings of the 10th
  International Conference on Theorem Proving in Higher Order Logics ({TPHOLs}
  '97)}, volume 1275 of {\em LNCS}, pages 307--322, Murray Hill, New Jersey,
  1997.

\end{thebibliography}

\end{document}